\documentclass[pdftex,11pt,svgnames]{article}

\usepackage[utf8]{inputenc}            
\usepackage{amsmath}
\usepackage{amssymb}
\usepackage{bm}  
\usepackage{amsbsy} 
\usepackage{listings}
\usepackage{graphics}
\usepackage{float}

\usepackage[multiple]{footmisc}

\usepackage{enumerate}
\usepackage[shortlabels]{enumitem}

\IfFileExists{dsfont.sty}{
	\usepackage{dsfont}
	\usepackage[usenames,dvipsnames]{color} 
	
}{
	
}
\newcommand{\ip}[2]{\langle #1 , #2 \rangle}
\usepackage[numbers]{natbib}

\usepackage{pdfpages}
\usepackage{hyperref} 
\hypersetup{
    unicode=false,          
    pdftoolbar=true,        
    pdfmenubar=true,        
    pdffitwindow=false,     
    pdfstartview={FitH},    
    pdftitle={Multivariate higher order moments in multi-state life insurance},    
    pdfauthor={Jamaal Ahmad},     
    pdfsubject={Life Insurance Mathematics},   
    pdfcreator={ },   
    pdfproducer={ }, 
    pdfkeywords={ } { } { }, 
    pdfnewwindow=true,      
    colorlinks=true,       
    linkcolor=Black, 
    citecolor=Black,        
    filecolor=Black,      
    urlcolor=Black 
}

\usepackage[font=small,format=plain,labelfont=bf,up,textfont=it,up]{caption}

\usepackage[a4paper]{geometry}  
\usepackage[english]{babel}

\usepackage[autolanguage,np]{numprint}

\usepackage{doi}

\usepackage{fancyhdr}
\usepackage{amsthm}
\theoremstyle{plain}
\newtheorem{Theorem}{Theorem}
\newtheorem{theorem}[Theorem]{Theorem}

\newtheorem{assumption}[Theorem]{Assumption}

\newtheorem{corollary}[Theorem]{Corollary}

\theoremstyle{definition}

\newtheorem{example}[Theorem]{Example}

\newenvironment{sproof}{%
  \proof}{\endproof}

\newcommand\xqed[1]{%
  \leavevmode\unskip\penalty9999 \hbox{}\nobreak\hfill
  \quad\hbox{#1}}
\newcommand\demo{\xqed{$\circ$}}
\newcommand\demoo{\xqed{$\triangle$}}

\theoremstyle{remark}

\newtheorem{remark}[Theorem]{Remark}

\numberwithin{Theorem}{section}

\numberwithin{equation}{section}

\newcommand{\dd}{{\, \mathrm{d} }}

\newcommand{\tuborg}[1]{\left\{ #1 \right\}}
\DeclareMathOperator{\Var}{Var}
\DeclareMathOperator{\E}{\mathbb{E}}

\DeclareMathOperator{\Cov}{Cov}

\newcommand{\1}[1]{\mathds{1}_{( #1)}}

\DeclareMathAlphabet{\mathpzc}{OMS}{pzc}{m}{it}
\newcommand{\J}{\mathpzc{J}}

\newcommand{\R}{\mathbb{R}}
\newcommand{\N}{\mathbb{N}}

\makeatletter

\newcommand{\Rmnum}[1]{\expandafter\@slowromancap\romannumeral #1@}
\makeatother

\usepackage{wasysym} 

\usepackage{tikz}
\usetikzlibrary{arrows,positioning,calc} 
\tikzset{
    >=stealth',
    punkt/.style={
           rectangle,
           rounded corners,
           draw=black, thick,
           text width=5em,
           minimum height=2em,
           text centered},
    punktl/.style={
           rectangle,
           rounded corners,
           draw=black, thick,
           text width=7em,
           minimum height=2em,
           text centered},
    pil/.style={
           ->,
           shorten <=4pt,
           shorten >=4pt,},
    pildotted/.style={
           ->,
           shorten <=4pt,
           shorten >=4pt,
  dotted,
  }
}
\usepackage{subfig}
\usepackage{todonotes}

\usepackage{authblk}

\makeatletter
\let\@fnsymbol\@arabic
\makeatother

\title{Multivariate higher order moments in multi-state life insurance}
\author[1,$\star$]{Jamaal Ahmad}
\affil[1]{\footnotesize Department of Mathematical Sciences, University of Copenhagen, Universitetsparken 5, DK-2100 Copenhagen \O, Denmark.}
\affil[$\star$]{\footnotesize Corresponding author. E-mail: \href{mailto:jamaal@math.ku.dk}{jamaal@math.ku.dk}.}

\date{ }

\begin{document}
\maketitle 

\begin{center}
{\sc Abstract}
\end{center}
{\small
It is well-known that combining life annuities and death benefits introduce opposite effects in payments with respect to the mortality risk on the lifetime of the insured. In a general multi-state framework with multiple product types, such joint effects are less trivial. In this paper, we consider a multivariate payment process in multi-state life insurance, where the components are defined in terms of the same Markovian state process. The multivariate present value of future payments is introduced, and we derive differential equations and product integral representations of its conditional moments and moment generating function. Special attention is given to pair-wise covariances between two present values, where results closely connected to Hattendorff type of results for the variance are derived. The results are illustrated in a numerical example in a disability model.   

\vspace{5mm}


\vspace{5mm}


}

\section{Introduction}
In this paper, we extend the results of \cite{Bladt2020} to multivariate payment processes with components defined in terms of the same multi-state Markov process. Our main contributions are differential equations and product integral representations of higher order moments of the multivariate present value, which are derived via its moment generating function. The main results appear as natural multivariate generalizations to those of \cite{Bladt2020, NorbergHigherOrder}, and to some extent also \cite{AdekambiChristiansen2017}, and this is pointed out in a series of remarks throughout the paper. We give special attention to pair-wise covariances between two present values and derive results that reveal close connections to Hattendorff type of results for the variance of a present value.   
 
The paper is motivated by the following. The use of multi-state time-inhomogeneous Markov models in life insurance, dating back to at least \cite{hoem69}, have provided a unified mathematical framework to model the random pattern of states of the insured with different kinds of life and health events; see e.g.\ an overview in \cite[Chapter V.2]{asmussensteffensen}.\ While these are models for primarily unsystematic biometric risk, one may also use multi-state Markov models to integrate systematic risk as in \cite{Norberg1995, NorbergBonus2}.\ For valuation and risk management, the present value of future payments is the main quantity of interest, and the conditional expected present value constitute the prospective reserve; this principle relies on diversification of unsystematic risk resulting from the law of large numbers on independent and identically distributed present values. The prospective reserve satisfies the celebrated Thiele's differential equations, cf.\ \cite{hoem69}.\ 

Since there may be a significant risk of future payments deviating from expected values, the insurer adds safety margins to premiums and reserves. While this is usually done implicitly via a first order basis consisting of prudent assumptions on interest and transition rates, an alternative approach is to compute safety margins explicitly via (properties of) the probability distribution of future payments. This is the focal point of \cite{christiansen2013}, where approximations based on the Central Limit Theorem (CLT) are established. Here, the variance of future payments is needed; it is obtained as an integral expression in \cite{hoem69}, see also \cite{hoemaalen, norberg1991}, and Hattendorff's theorem, which is formulated in the multi-state framework by \cite{Ramlau1988, norberg1992}, provides particularly simple formulae in terms of the associated multivariate counting process and so-called sum at risks. Turning to higher order moments beyond the variance increases the precision in these kinds of approximations; in \cite{NorbergHigherOrder}, we have differential equations for all moments of future payments, and in \cite{Bladt2020}, a general matrix-based framework using product integrals is developed to compute said moments, from which densities and distribution functions of future payments are approximated via polynomial expansions. These type of results are extended to the semi-Markovian framework in \cite{AdekambiChristiansen2017} to allow for duration dependency in payments and transition rates. Other ways of calculating distribution functions can be found in \cite{HesselagerNorberg, AdekambiChristiansen2020}, where integral and differential equations for these are derived.   

In practice, an insured typically holds a combination of various product types in order to be covered in case of different kinds of life and health events. In the multi-state framework, this is without further notice handled by considering the aggregated payment process, and reserves, higher order moments and probability distributions can be calculated using the methods outlined above. However, since these product types are contingent on events that work in different directions w.r.t.\ their underlying risks, a decomposition of the total risk into the different risk types is important in pricing and reserving as well as in risk management. In \cite{christiansen2013}, asymptotic safety margins are obtained for premiums and reserves, from which a decomposition of risk types corresponding to the different transitions in the Markov chain is established. In this paper, we take a different point of view; we decompose the payment process into different product types and examine the joint distribution of their future payments by calculating its moments; in particular, we calculate covariance matrices and illustrate how they may be used to approximate joint safety margins via multivariate CLT approximations. The simplest example is the product combination of life annuities and death benefits, which have opposite effects with respect to the mortality risk on the lifetime of the insured. In the general multi-state framework of this paper, more complex and non-trivial interactions can be examined; we give a motivating example in a disability model in the next section.    

For systematic risk, the idea of mixing various product types and analysing their interacting effect has already been discussed extensively in the concept of so-called natural hedging in life insurance, where one considers the effects on the reserves with respect to future changes in mortality rates in a portfolio consisting exactly of life annuities and death benefits; see e.g.\ \cite{coxlin}.\ While this concept arose in the survival model at first, an example in the multi-state model have appeared recently in \cite{italian} where the authors analyse interactions with life products and long-term care insurance using a disability model. Results from this paper allows us to carry out these kinds of analyses in a general multi-state framework, however with models for systematic risk restricted to time-inhomogeneous Markov chains on finite state spaces.         
    
The paper is structured as follows. We start out in Section \ref{sec:motivation} with presenting a motivating example of decomposition of payment processes into different types. Section \ref{sec:setup} then introduces the setup of the paper:\ the multi-state Markov process governing the state of the insured, the multivariate payment process describing the collection of life insurance contracts held by the insured, and the multivariate present value of their payments. In Section \ref{sec:mgf}, we derive product integral representations of the conditional moment generating function of the present value, which is shown to satisfy a system of partial differential equations. In Section \ref{sec:moments} we use this to derive ordinary differential equations and product integral representations for the conditional moments. Finally, in Section \ref{sec:numerical} we illustrate results of this paper in a numerical example of the motivating example from Section \ref{sec:motivation}.  
\newpage
 \section{Motivating example}\label{sec:motivation}
Before presenting the general setup and main results of the paper, we start out with giving a motivating example of the  problems we wish to solve. The setup of the example is a disability model, which serves to illustrate the motivations in the simplest non-trivial multi-state life insurance setting. 

Consider a single insured having a deterministic retirement time $T$ with the following product combination:
\begin{enumerate}[(1), leftmargin = 1cm]
\item Death benefit paying $S$ upon death before time $T$.    
\item Deferred life annuity paying a benefit rate of $b$ while alive, starting from  time $T$. 
\item Disability annuity paying a benefit rate of $d$ while disabled until time $T$.
\end{enumerate}  
The state of the insured is modelled in the disability model with recoveries, as depicted in Figure \ref{fig:invalidemodel}.\  
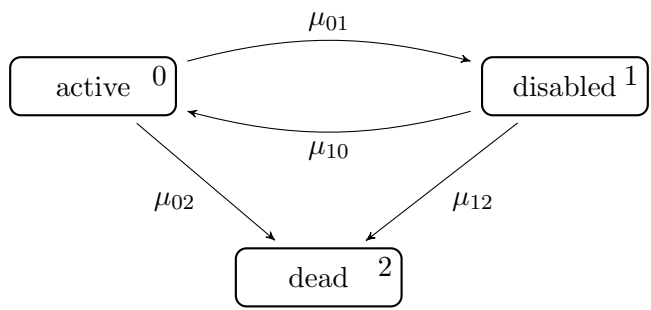
\begin{figure}[h]
	\centering
	\begin{tikzpicture}[node distance=2em and 0em]
		\node[punkt] (1) {disabled};
		\node[anchor=north east, at=(1.north east)]{$1$};
		\node[punkt, left = 40mm of 1] (0) {active};
		\node[anchor=north east, at=(0.north east)]{$0$};
		\node[draw = none, fill = none, left = 20 mm of 1] (test) {};
		\node[punkt, below = 20mm of test] (2) {dead};
		\node[anchor=north east, at=(2.north east)]{$2$};
	\path
		(0)	edge [pil, bend left = 15]	node [above]			{$\mu_{01}$}		(1)
		(1)	edge [pil]				node [below right]	{$\mu_{12}$}		(2)
		(0)	edge [pil]				node [below left]		{$\mu_{02}$}		(2)
		(1)	edge [pil, bend left = 15]	node [below]			{$\mu_{10}$}		(0)
	;
	\end{tikzpicture}
	\caption{Disability model with recoveries.}
	\label{fig:invalidemodel}
\end{figure} 

The payment process $B$ giving accumulated benefits from this product design then takes the form 
\begin{align*}
\dd B(t) &= \mathds{1}_{(Z(t-)=0)}\!\left(b\cdot \mathds{1}_{(t\geq T)}\dd t + S\cdot \mathds{1}_{(t<T)}\dd N_{02}(t)\right)\\[0.2 cm]
&\quad + \mathds{1}_{(Z(t-)=1)}\!\left(d\cdot \mathds{1}_{(t<T)}\dd t + b\cdot \mathds{1}_{(t\geq T)}\dd t + S\cdot \mathds{1}_{(t<T)}\dd N_{12}(t)\right), \quad B(0)=0,  
\end{align*} 
where $Z = \tuborg{Z(t)}_{t\geq 0}$ is the process (with state-space according to Figure \ref{fig:invalidemodel}) indicating the state of the insured and $N_{ij}$ is the associated counting process counting the number of jumps from state $i$ to $j$, for $i,j\in \tuborg{0,1,2}$, $i\neq j$.

The present value at initiation of the contract (at time 0) of future benefits up to some terminal time point $t>0$ is then given by
\begin{align*}
U(0,t) = \int_0^t e^{-\int_0^s r(v)\dd v} \dd B(s),
\end{align*}   
where $r = \{r(v)\}_{v\geq 0}$ is some deterministic interest rate that accounts for time value of money. Based on a principle of diversification of unsystematic biometric risk, the prospective reserve is given as the expected present value, i.e.\ $V(0,t) = \E\!\left[U(0,t)\right]$, which satisfies the well-known Thiele's differential equation, cf.\ \cite{hoem69}.\ This principle relies on the strong law of large numbers on i.i.d.\ present values (coming from i.i.d.\ insured with same product combination), $U^1(0,t),U^2(0,t),\ldots$, 
 \begin{align*}
 \frac{1}{N}\sum_{i=1}^N U^i(0,t) \overset{\ \text{a.s } \ }{\longrightarrow} \E\!\left[U^1(0,t)\right]\!, \quad\text{as}\ \ N\rightarrow\infty. 
 \end{align*}
To avoid a significant risk of future payments deviating from expected values, the insurer adds safety margins to the prospective reserve. Assuming an explicit approach based on the probability distribution of the present value, a natural approach is to use approximations based on the Central Limit Theorem (CLT), cf.\ \cite{christiansen2013}: 
\begin{align*}
\sqrt{N}\!\left(\frac{1}{N}\sum_{i=1}^N U^i(0,t) - \E\!\left[U^1(0,t)\right]\right) \overset{\mathcal{D}}{\longrightarrow} \mathcal{N}\!\left(0,\Var\!\left[U^1(0,t)\right]\right)\!, \ \ \text{as}\ \ N\rightarrow\infty.
\end{align*}
We recall that $\Var\!\left[U^1(0,t)\right]$ can be calculated using multi-state versions of Hattendorff's theorem, cf.\ \cite{Ramlau1988}. One may even turn to higher order moments beyond the variance to approximate these distributions further, cf.\ \cite{NorbergHigherOrder, Bladt2020}.      

Now, since the different payment types work in different directions w.r.t.\ the underlying mortality and disability risk, the safety margins on the individual product types are dependent. We are therefore interested in computing joint safety margins that suitably represents decompositions of these type of risks. This we can obtain by decomposing the payment process $B$ into the three types corresponding to $(1)-(3)$ above: 
\begin{alignat*}{2}
B(t) &= B_1(t) + B_2(t) + B_3(t), && \\[0.2 cm]
\dd B_1(t) &= S\cdot \mathds{1}_{(t<T)}\!\left(\mathds{1}_{(Z(t-)=0)}\dd N_{02}(t) + \mathds{1}_{(Z(t-)=1)} \dd N_{12}(t)\right),  \quad &&B_1(0) = 0, \\[0.2 cm]
\dd B_2(t) &= b\cdot \mathds{1}_{(t\geq T)}\!\left(\mathds{1}_{(Z(t-)=0)} + \mathds{1}_{(Z(t-)=1)} \right)\!\dd t,  &&B_2(0) = 0, \\[0.2 cm]
\dd B_3(t) &= d\cdot \mathds{1}_{(t<T)}\mathds{1}_{(Z(t-)=1)}\dd t, &&B_3(0) = 0,
\end{alignat*}
and where the present value of the payment processes $B_1,B_2$ and $B_3$ are coupled in a vector: $\bm{U}(0,t) = \left(U_1(0,t),U_2(0,t),U_3(0,t)\right)\!'$.\ The joint distribution of $\bm{U}(0,t)$ now becomes the key to obtain joint safety margins related to the different product types. With i.i.d. present values $\bm{U}^1(0,t),\bm{U}^2(0,t),\ldots$, the multivariate CLT now applies     
\begin{align*}
&\sqrt{N}\!\left(\frac{1}{N}\sum_{i=1}^N \bm{U}^i(0,t) - \E\!\left[\bm{U}^1(0,t)\right]\right) \overset{\mathcal{D}}{\longrightarrow} \mathcal{N}\!\left(\bm{0},\bm{\Sigma}_{\bm{U}^1(0,t)} \right)\!, \ \ \text{as}\ \ N\rightarrow\infty, \\[0.4 cm]
&\bm{\Sigma}_{\bm{U}^1(0,t)} = \left\{\Cov\!\left(U^1_i(0,t),U^1_j(0,t)\right) \right\}_{i,j=1,2,3}.
\end{align*}
Hence, the covariance matrix of $\bm{U}^1(0,t)$ allows for approximation of joint safety margins for the three product types. Again, turning to higher order moments of $\bm{U}^1(0,t)$ in general improves these approximations in the same way as for the univariate case. The focal point of this paper is to compute these higher order moments for multivariate present values in a general multi-state framework, and give special attention to pair-wise covariances as motivated by this example. 

We should like to mention that while this example concerned joint safety margins on the unsystematic biometric risk in life and disability insurance, one could have used the same model for e.g.\ systematic mortality risk. This is obtained by thinking of the states as demographic states, say, where e.g.\ the disability state represents a situation where the overall mortality for all insured in the portfolio have increased (or decreased). Then we would obtain joint safety margins related to systematic mortality risk when considering the joint distribution of the multivariate present value (a modification of the payment processes might be necessary here).

\section{Setup}\label{sec:setup}
We now proceed with the general setup of the paper. Subsection \ref{sec:prelim} introduces the process governing the state of the insured and the multivariate payment process describing the collection of products held by the insured. Subsection \ref{sec:valuation} then considers the multivariate present value of future payments, whose components consists of the present value of the individual payments, and we introduce the calculation of their higher order moments as the main purpose of the paper. By considering existing results in \cite{NorbergHigherOrder, Bladt2020} for single payment processes, we end the section by explaining our main ideas and give an overview on how we are to obtain the multivariate extension in this paper.

\textbf{Notation and conventions} \ For vectors $\bm{x},\bm{y}\in \R^n$,  we say that $\bm{x}\leq \bm{y}$ if $x_\ell \leq y_\ell$ for all $\ell=1,\ldots,n$. Furthermore, we denote with $\bar{\bm{x}}$ the sum of the elements in $\bm{x}$, i.e.\ 
$
\bar{\bm{x}} = \sum_{\ell=1}^n x_\ell 
$.\ Also, $\bm{e}_\ell$ is the $\ell$'th unit vector in $\R^n$ and 
$
E_n = \tuborg{\bm{e}_1,\ldots,\bm{e}_n}
$
denotes the set of unit vectors, i.e.\ the natural basis of $\R^n$.\ Furthermore, $\ip{\bm{x}}{\bm{y}} = \sum_{\ell=1}^n x_\ell y_\ell$ denotes the usual inner product. 

For $n\times n$ matrix valued functions $\R\ni x\mapsto \bm{A}(x)$ we denote its product integral in the (time) interval ${\color{blue} ( }s,t]$ as  
\begin{align*}
\bm{F}(s,t) = \prod_s^t \left(\bm{I}+\bm{A}(x)\dd x\right)\!,
\end{align*}  
whenever it exists; here $\bm{I}$ is the $n\times n$ identity matrix. For a survey on product integration and its properties, we refer to \cite{GillJohansen, Johansen} and the applications to life insurance in \cite{milbrodtstracke1997, Bladt2020}. Numerical schemes for calculation of product integrals can be found in \cite{HeltonStucwisch}.  

For the purpose of computing multivariate higher order moments, we consider the map $S:\N_0^n \rightarrow \mathpzc{P}(\N_0^n)$ defined as
\begin{equation}\label{eq:Sk}
S({\bm{x}}) = \left\{ \bm{\xi}\in \N_0^n\setminus\tuborg{\bm{0}} \, \big| \, \bm{\xi}\leq \bm{x}\right\}\!, \quad \bm{x}\in \N_0^n. 
\end{equation}
The set $S(\bm{x})$ then contains all combinations of lower order moments of $\bm{x}$, excluding the $\bm{0}$'th moment, and the cardinality $\big|S({\bm{x}})\big| = \prod_{\ell=1}^n (x_\ell+1)-1$ then represents the amount of these that appear; whenever we wish to include the $\bm{0}$'th moment, we write $\widetilde{S}(\bm{k}) = S(\bm{k})\cup \tuborg{\bm{0}}$.\ For notational convenience, we simply define the mapping $|\cdot|:\N_0^n \rightarrow \N_0$ as 
\begin{equation*}
|\bm{x}| := \prod_{\ell=1}^n (x_\ell+1)-1, \quad \bm{x}\in \N_0^n. 
\end{equation*}
To compute higher order partial derivatives of multivariate scalar functions $f:\R^n\rightarrow \R$, we use the following notation, for $\bm{m}\in \N_0^n$,
\begin{align*}
\frac{\partial^{\bm{m}}}{\partial\bm{x}^{\bm{m}}} f(\bm{x}) = \frac{\partial^{\bar{\bm{m}}}}{\partial x_1^{m_1}\partial x_2^{m_2}\cdot\ldots\cdot \partial x_n^{m_n}} f(\bm{x})
\end{align*}
provided they exist.

\subsection{Life insurance model}\label{sec:prelim}
The state of the insured is governed by a non-explosive Markov jump process $Z=\tuborg{Z(t)}_{t\geq 0}$ taking values on a finite state space $\J=\{0,\ldots,J-1\}$, $J\in \N$, with deterministic initial state $Z(0)=z_0\in \J$.\ Denote with $\bm{M}(t) = \tuborg{\mu_{ij}(t)}_{i,j\in \J}$ the transition intensity matrix of $Z$ and let $\bm{P}(s,t) = \tuborg{p_{ij}(s,t)}_{i,j\in \J}$ be the corresponding transition probability matrix. The multivariate counting process associated with $Z$ is denoted $\bm{N}(t)=\tuborg{N_{ij}(t)}_{i,j\in \J, i\neq j}$, with components given by
\begin{equation*}
N_{ij}(t) = \#\tuborg{s\in (0,t] : Z(s-) = i, Z(s) = j}.
\end{equation*} 
We consider a collection of life insurance contracts described by the $n$-dimensional payment process $\bm{B}(t) = \left(B_1(t),\ldots,B_n(t) \right)'$ with each component giving accumulated benefits less premiums for the corresponding contract; here $n\in \N$ is a fixed and finite number of contracts. We suppose that the $\ell$'th payment process, $\ell=1,\ldots,n$, consist of deterministic sojourn payment rates $b_i^\ell$ and transition payments $b_{ij}^\ell$, $i,j\in \J, i\neq j$, in terms of the state process $Z$.\ Coupling these in vectors, 
\begin{align}\label{eq:b_j}
\bm{b}_i(t) &= \left(b_i^1(t),\ldots,b_i^n(t)\right)', \\\label{eq:b_jk}
\bm{b}_{ij}(t) &= (b^1_{ij}(t),\ldots,b^n_{ij}(t))',
\end{align}
we formally characterize $\bm{B}$ as
\begin{equation} \label{eq:payment_pr}
\dd \bm{B}(t) = \sum_{i\in \J} \mathds{1}_{(Z(t-)=i)}\bigg( \bm{b}_{i}(t)\!\dd t + \sum_{j:j\neq i} \bm{b}_{ij}(t)\!\dd N_{ij}(t)\bigg), \quad \bm{B}(0)=\bm{0}. 
\end{equation}
We assume throughout that the transition rates $\mu_{ij}$ as well as the payment functions $b^\ell_i$, $b_{ij}^\ell$, for $i,j\in \J$, $i\neq j$, and $\ell=1,\ldots,n$, are bounded on finite intervals. This makes all integrals and expectations in the following well defined and finite.  

We think of $\bm{B}$ as describing payments with different product types that naturally are dependent. The simplest example is the combination of life annuities and death benefit, which have opposite effects with respect to the mortality risk on lifetimes.\ In the general multi-state framework of this paper, we are able to examine product interactions that are less trivial; cf.\ the motivating example in Section \ref{sec:motivation}.
\begin{example}[Disability model]\label{ex:main_payment_functions}
The setup of the motivating example from Section \ref{sec:motivation} in the disability model corresponds to having the   following vectors of payments functions \eqref{eq:b_j}--\eqref{eq:b_jk}:
\begin{align*}
\bm{b}_0(t) &= \left(0, b\cdot \mathds{1}_{(t\geq T)}, 0\right)', \\
\bm{b}_1(t) &= \left(0, b\cdot \mathds{1}_{(t\geq T)}, d\cdot \mathds{1}_{(t<T)} \right)', \\
\bm{b}_{02}(t)=\bm{b}_{12}(t) &=  \left(S\cdot\mathds{1}_{(t<T)}, 0, 0  \right)',  
\end{align*}
with the remaining vectors of payment functions having purely zero entries. \demo
\end{example}
\subsection{Present values and their moments}\label{sec:valuation}
Let $r=\tuborg{r(t)}_{t\geq 0}$ be a continuous, suitably regular and deterministic short rate.\ To account for time value of money, introduce $v(s,t)$ as the discount factor for a payment at time $t$ valuated at time $s\leq t$, i.e.\  
\begin{align*}
v(s,t) = e^{-\int_s^t r(x)\dd x}. 
\end{align*}
The present value at time $s$ of future benefits less premiums up to time $t$ for the payment process $\bm{B}$ is then given by 
\begin{align}\label{eq:multi_U}
\bm{U}(s,t)   &= \int_s^t v(s,u) \dd \bm{B}(u) = \left(U_1(s,t),\ldots,U_n(s,t)\right)',
\end{align}
where the components $U_\ell(s,t) = \int_s^t v(s,u) \dd B_\ell(u)$, $\ell=1,\ldots,n$, are the present values of the individual payment process. 

The components of $\bm{U}$ are (by construction) seen to be dependent, which introduces joint effects on the payments with respect to the  underlying event risk. To analyse these dependence structures, we wish to examine its joint distribution
by computing its conditional higher order moments given the current state of the insured, that is, by computing
\begin{align}\label{eq:Vi}
V^{(\bm{k})}_i(s,t) &= \E\!\left[\left. \, \prod_{\ell = 1}^n U_\ell(s,t)^{k_\ell} \, \right| Z(s) = i\right]\!
\end{align}
for any $i\in \J$ and $\bm{k}=(k_1,\ldots,k_n)\in \N_0^n$.\ For $\bm{k} = \bm{e}_\ell$ we have that $V_i^{(\bm{e}_\ell)}(s,t)$ is the state-wise prospective reserve for the $\ell$'th payment process, calculated using the classic Thiele's differential equation, cf.\ \cite{hoem69}.    

 As mentioned in the introduction, having multiple product types for the same insured is usually handled by adding all payment functions together as $\bar{\bm{b}}_i$ and $\bar{\bm{b}}_{ij}$, or, equivalently, by considering the aggregate present value $\bar{\bm{U}}(s,t)$.\ This effectively gives a one dimensional payment process to examine; in \cite{hoem69} we have Thiele's differential equation for conditional expectations and in \cite{AdekambiChristiansen2017,  Bladt2020,NorbergHigherOrder} we have differential equations for conditional higher order moments, which e.g.\ allows for calculation of safety margins on reserves as in \cite{christiansen2013}. The present paper is therefore to be thought of as a way of decomposing the aggregated payments into the different product types and study the joint distribution of their future payments, which then would allow for calculation of joint safety margins for the different product types as outlined in Section \ref{sec:motivation}.

The following matrices, which we think of as inputs to the computations, are needed in the derivations; for $\ell \in \{ 1,\ldots,n\}$, $i,j\in \J, i\neq j$, and $\bm{y}\in \N_0^n\setminus (E_n\cup \tuborg{\bm{0}})$, let
\begin{align}\label{eq:matrix_first}
\bm{B}_\ell(s) &= \left\{b^\ell_{ij}(s)\right\}_{i,j\in \J}, \\
\bm{b}^\ell(s) &= \left(b^\ell_0(s),\ldots,b^\ell_{J+1}(s)\right)', \\ 
 \label{eq:matrix_mellem_plusone}
\bm{R}_\ell (s) &= \bm{M}(s)\bullet \bm{B}_\ell(s) + \bm{\Delta}(\bm{b}^\ell(s)), \\ \label{eq:matrix_last}
\bm{C}^{(\bm{y})}(s) &= \bm{M}(s)\bullet \bm{B}^{\bullet y_1}_1(s)\bullet \ldots\bullet \bm{B}_n^{\bullet y_n}(s),
\end{align}
where $\Delta(\bm{b}^\ell(s))$ denotes the diagonal matrix with the vector $\bm{b}^\ell(s)$ as diagonal elements, $\bullet$ denotes the entry-wise matrix product and e.g.\ $\bm{B}_1^{\bullet y}(s)$ denotes the entry-wise matrix product of $\bm{B}_1(s)$ with itself $y$ times; here we use the convention that e.g.\ $\bm{B}_1^{\bullet 0}(s)$ denotes the matrix with all entries equal to 1.\ These matrices are defined analogously to those of \cite[(3.8)--(3.11)]{Bladt2020} with the relevant extension to a multivariate payment process.  

The remainder of the paper now focuses on computation of $\bm{V}_i^{(\bm{k})}(s,t)$, carried out as follows. Following \cite{Bladt2020}, we first compute the conditional higher order moments on events that the insured terminates in specific states, given by, for $j\in \J$,   
\begin{align}\label{eq:Vij}
V_{ij}^{(\bm{k})}(s,t) = \E\!\left[\left. \mathds{1}_{(Z(t)=j)}\prod_{\ell = 1}^n U_\ell(s,t)^{k_\ell} \, \right| Z(s) = i\right]
\end{align} 
with corresponding matrix
\begin{align}\label{eq:V}
\bm{V}^{(\bm{k})}(s,t) = \left\{V_{ij}^{(\bm{k})}(s,t)\right\}_{i,j\in \J}.
\end{align}
This is carried out by deriving the relevant moment generating function.\ As in \cite{Bladt2020}, we shall refer to these as the \textit{conditional partial moments} of $\bm{U}(s,t)$.\ From these, the conditional moments $V_{i}^{(\bm{k})}(s,t)$ are obtained via the relation $\sum_{j\in \J} V_{ij}^{(\bm{k})}(s,t) = V_{i}^{(\bm{k})}(s,t)$,  such that
\begin{align}\label{eq:Vi_expr}
V_{i}^{(\bm{k})}(s,t) = \tuborg{\bm{V}^{(\bm{k})}(s,t)\bm{1}}_i, 
\end{align}  
where $\bm{1}=(1,1,\ldots,1)'$.\  

For the first order moment $k=1$ (in the one-dimensional case $n=1$), the partial reserves \eqref{eq:Vij} have already been introduced in existing literature; they coincide with the conditional prospective premium reserves introduced by \cite{Wolthuis1992}.\ In \cite{norberg1991}, we also encounter these type of quantities for corresponding retrospective reserves initiated at time zero, i.e.\ for $s=0$.

These earlier treatments in the literature of these kinds of quantities justifies the conditional partial moments \eqref{eq:V} as relevant objects to study on their own. In this paper, however, they are primarily introduced as a mathematical convenient object towards obtaining the conditional moments \eqref{eq:Vi}; this is the same focus in \cite{Bladt2020}.\

\section{Moment generating functions}\label{sec:mgf}
Consider the multivariate present value $\bm{U}(s,t)$. The moment generating function of its conditional distribution given $Z(s) = i$ on the event $(Z(t)=j)$, $i,j\in \J$,  is given by
\begin{align}\label{eq:Fij_def}
F_{ij}(\bm{\theta}; s,t) &= \E\!\left[\left. e^{\ip{\bm{\theta}}{ \bm{U}(s,t) } }\mathds{1}_{(Z(t)=j)}\, \right| Z(s)=i\right]\!, \quad \bm{\theta} \in \R^n,
\end{align}
with corresponding matrix
\begin{align*}
\bm{F}(\bm{\theta}; s,t) &= \tuborg{F_{ij}(\bm{\theta}; s,t)}_{i,j\in \J}.
\end{align*}
We are then able to obtain the moment generating function by a direct application of results in \cite{Bladt2020}.
\begin{theorem}\label{thm:mgf}
The joint distribution of the multivariate present value $\bm{U}(s,t)$ has  moment generating function given by 
\begin{align*}
\bm{F}(\bm{\theta}; s,t) = \prod_s^t \left(\bm{I} + \left[\bm{M}(u)\bullet \tuborg{e^{v(s,u)\ip{ \bm{\theta} }{ \bm{b}_{ij}(u)}}}_{i,j\in \J}+v(s,u)\sum_{\ell = 1}^n \theta_\ell \bm{\Delta}(\bm{b}^\ell(u))  \right]\!\!\dd u \right) 
\end{align*}
\end{theorem}
\begin{proof}
Let $\bm{\theta}\in \R^n$ and $0\leq s\leq t$ be given. By linearity of integrals, we note 
\begin{align*}
\ip{\bm{\theta}}{\bm{U}(s,t)} &= \int_s^t v(s,u) \dd \widetilde{B}_{\bm{\theta}}(u), \\
\dd \widetilde{B}_{\bm{\theta}}(u) &=
\sum_{i\in \J}\1{Z(u-)=i}\bigg(\ip{\bm{\theta}}{ \bm{b}_{i}(u) }\!
 \dd u + \sum_{j:j\neq i} \ip{\bm{\theta}}{\bm{b}_{ij}(u)}\!\dd N_{ij}(u)\bigg).
\end{align*}
Hence, $\ip{\bm{\theta}}{ \bm{U}(s,t) }$ is the present value of a one-dimensional payment process with sojourn payment rates $\ip{\bm{\theta}}{ \bm{b}_j(\cdot) }$ and transition payments $\ip{\bm{\theta}}{ \bm{b}_{jk}(\cdot)}$.\ Thus, in the spirit of \cite{Bladt2020}, we may view $\bm{F}(\bm{\theta}; s,t)$ as the moment generating function of a total (undiscounted) reward evaluated in 1 with the payment functions $u\mapsto \ip{\bm{\theta}}{ v(s,u)\bm{b}_j(u) } $ and $u\mapsto \ip{\bm{\theta}}{ v(s,u)\bm{b}_{jk}(u) }$ (for fixed $s$).\ The result then follows from an application of \cite[Theorem 3]{Bladt2020}.    
\end{proof}
\begin{remark}
Theorem \ref{thm:mgf} is a generalization of \cite[Theorem 3]{Bladt2020} to multivariate payment processes with discounting.\ Indeed, by letting $n=1$ and $r(s) = 0$ for all $s\geq 0$, the result simplifies to, suppressing the superfluous dependency in $\ell$, 
\begin{align*}
\bm{F}(\theta; s,t) = \prod_s^t \left( \bm{I} + \left[\bm{M}(s)\bullet \tuborg{e^{\theta b_{ij}(s)}}_{i,j\in \J}+\theta \bm{\Delta}(\bm{b}(s))  \right]\!\!\dd u\right)\!,
\end{align*}  
which is the product integral of \cite[Theorem 3]{Bladt2020}. \demoo
\end{remark}
From the product integral representation of $\bm{F}$, we are able to obtain a partial differential equation satisfied by $\bm{F}$. Similar type of equations are also obtained in \cite{AdekambiChristiansen2017}. 
\begin{theorem}\label{thm:mgf_pde}
The moment generating function $\bm{F}$ satisfies the system of partial differential equations 
\begin{equation}
\begin{aligned} \label{eq:mgf_pde}
\frac{\partial}{\partial s} \bm{F}(\bm{\theta}; s,t) &= r(s)\sum_{\ell=1}^n \theta_\ell \frac{\partial }{\partial \theta_\ell}\bm{F}(\bm{\theta}; s,t) \\ &\quad -\left[\bm{M}(s)\bullet \tuborg{e^{\ip{\bm{\theta}}{\bm{b}_{ij}(s)}}}_{i,j\in \J}+\sum_{\ell  = 1}^n \theta_\ell \bm{\Delta}(\bm{b}^\ell(s))  \right]\!\bm{F}(\bm{\theta}; s,t), 
\end{aligned}
\end{equation}
with the terminal condition $\bm{F}(\bm{\theta}; t,t) = \bm{I}$. 
\end{theorem}
\begin{proof}
For $\bm{\theta}\in \R^n$ and $0\leq s\leq u$, let 
\begin{align*}
\bm{A}(s,u;\bm{\theta}) = \bm{M}(u)\bullet \tuborg{e^{v(s,u)\ip{ \bm{\theta} }{ \bm{b}_{ij}(u)}}}_{i,j\in \J}+v(s,u)\sum_{\ell = 1}^n \theta_\ell \bm{\Delta}(\bm{b}^\ell(u))
\end{align*}
denote the matrix that is product integrated in Theorem \ref{thm:mgf}.\ It satisfies the relation $\frac{\partial}{\partial s}\bm{A}(s,u;\bm{\theta}) = r(s)\sum_{\ell = 1}^n \theta_\ell \frac{\partial}{\partial \theta_\ell}\bm{A}(s,u;\bm{\theta})$, which can be seen from direct calculations. Now, since $\bm{F}$ is a product integral of the matrix function $\bm{A}(s,\cdot;\bm{\theta})$, it has the Peano-Baker series representation given by
\begin{align*}
\bm{F}(\bm{\theta}; s, t) = \bm{I}+\sum_{n=1}^\infty \int_s^t \int_s^{x_n}\cdots \int_s^{x_2} \bm{A}(s,x_1;\bm{\theta})\cdots \bm{A}(s,x_n;\bm{\theta})\dd x_1\cdots\dd x_n.
\end{align*}
Differentiating this w.r.t.\ $s$, we get using Lebniz' integral rule, 
\begin{align*}
\frac{\partial}{\partial s} \bm{F}(\bm{\theta}; s, t) &= -\bm{A}(s,s;\bm{\theta})\bm{F}(\bm{\theta}; s, t) \\
&\quad +\sum_{n=1}^\infty \int_s^t \int_s^{x_n}\cdots \int_s^{x_2} \frac{\partial}{\partial s}\left(\bm{A}(s,x_1;\bm{\theta})\cdots \bm{A}(s,x_n;\bm{\theta})\right)\!\dd x_1\cdots\dd x_n
\end{align*}
In the last differentiation, we can in each combination (when applying the product rule) substitute the derivative w.r.t.\ to $s$ with $r(s)\sum_{\ell = 1}^n \theta_\ell \frac{\partial}{\partial \theta_\ell}\bm{A}(s,u;\bm{\theta})$.\ This gives  
\begin{align*}
\frac{\partial}{\partial s} \bm{F}(\bm{\theta}; s, t) &= -\bm{A}(s,s;\bm{\theta})\bm{F}(\bm{\theta}; s, t) \\
& +\sum_{n=1}^\infty \int_s^t \int_s^{x_n}\cdots \int_s^{x_2}r(s)\!\sum_{\ell = 1}^n \theta_\ell \frac{\partial}{\partial \theta_\ell}\!\left(\bm{A}(s,x_1;\bm{\theta})\cdots \bm{A}(s,x_n;\bm{\theta})\right)\!\dd x_1\cdots\dd x_n \\[0.2 cm]
&= -\bm{A}(s,s;\bm{\theta})\bm{F}(\bm{\theta}; s, t) +r(s)\sum_{\ell = 1}^n \theta_\ell \frac{\partial}{\partial \theta_\ell}\bm{F}(\bm{\theta}; s, t),
\end{align*}
as claimed. The boundary condition follows from the definition of product integrals. Throughout the proof, we have used that the Peano-Baker series converges uniformly on compact intervals in order to interchange summation and differentiation.  
\end{proof}
\begin{remark}
Consider the univariate case $n=1$ and suppress the superfluous dependency in $\ell$.\ By multiplying both sides of \eqref{eq:mgf_pde} with $\bm{1} = (1,1,\ldots,1)'$ and extracting the $i$'th element, $i\in \J$, we see that $F_i(\theta; s,t) := \E\!\left[\left. e^{\theta U(s,t)} \right|  Z(s) = i\right]$ satisfies the partial differential equation  
\begin{align*}
\frac{\partial}{\partial s} F_i(\theta; s,t) &= r(s)\theta \frac{\partial}{\partial \theta} F_i(\theta; s,t) - \theta b_i(s)F_i(\theta; s,t) \\
&\quad -\sum_{k:k\neq i} \mu_{ik}(s)\left(e^{\theta b_{ik}(s) }F_{k}(\theta; s,t) - F_i(\theta;s,t)\right)\!, 
\end{align*}
with the boundary condition $F_i(\theta; t, t) = 1$.\ This is the partial differential equation of \cite[Proposition 4.3]{AdekambiChristiansen2017} in the Markovian special case with no duration dependency.  \demoo  
\end{remark} 
Having derived the relevant moment generating function of the multivariate present value $\bm{U}(s,t)$, we are now ready to compute the corresponding higher order moments. 

\section{Higher order moments}\label{sec:moments}
In this section, we derive the conditional (partial) higher order moments \eqref{eq:Vi} and \eqref{eq:V}, which constitutes the main contribution of the paper. In Subsection \ref{subsec:diff_moments} we derive differential equations for the conditional partial moments \eqref{eq:V} using the moment generating function $\bm{F}(\cdot; s,t)$ derived in Section \ref{sec:mgf}. Then we use this result to derive differential equations for the conditional moments \eqref{eq:Vi} via the relation \eqref{eq:Vi_expr}. In Subsection \ref{subsec:prodint_moments}, we then show how all conditional partial moments up to a given order $\bm{k}\in \N_0^n$ can be obtained from a product integral calculation. Throughout the section, we highlight how the results appear as natural and simple multivariate extensions to results in \cite{Bladt2020, NorbergHigherOrder}, in particular when notation may appear as cumbersome.           

\subsection{Differential equations of conditional (partial) moments}\label{subsec:diff_moments}
We now turn our attention to deriving ordinary differential equations for the conditional (partial) moments. For this, recall that $S(\bm{k})$ denotes the set of all lower order moments of the $\bm{k}$'th moment, excluding the $\bm{0}$'th moment, with elements we may write as
\begin{align*}
S(\bm{k}) = \tuborg{\bm{y}^1,\ldots,\bm{y}^{|\bm{k}|}}.
\end{align*}
The following main result then take use of the moment generating function derived in Section \ref{sec:mgf} to obtain differential equations for the conditional partial moments. 
\begin{theorem}\label{thm:diff_V}
The conditional partial moments $\bm{V}^{(\bm{k})}(\cdot,t)$ satisfies the backward differential equations given by, for $\bm{k}\in \N^n_0$, 
\begin{align}
\begin{split}\label{eq:diff_Vij}
\frac{\partial}{\partial s} \bm{V}^{(\bm{k})}(s,t) &= \left(\bar{\bm{k}}\, r(s)\bm{I} - \bm{M}(s)\right)\!\bm{V}^{(\bm{k})}(s,t) - \sum_{\ell=1}^n k_\ell\bm{R}_\ell(s)\bm{V}^{(\bm{k}-\bm{e}_\ell)}(s,t)  \\ 
&\quad - \sum_{\bm{y}\in S({\bm{k}}) \atop \bm{y}\notin  E_n}\prod_{\ell = 1}^n {k_\ell \choose y_\ell}\bm{C}^{(\bm{y})}(s)\bm{V}^{(\bm{k}-\bm{y})}(s,t),
\end{split}
\end{align}
with the terminal conditions $\bm{V}^{(\bm{k})}(t,t)=\mathds{1}_{(\bm{k}=\bm{0})}\bm{I}$.\ \end{theorem}
\begin{proof}
We can use the moment generating function  $\bm{F}(\cdot \, ; s,t)$ to calculate $\frac{\partial}{\partial s}\bm{V}^{(\bm{k})}(s,t)$ as 
\begin{align*}
 \frac{\partial}{\partial s} \bm{V}^{(\bm{k})}(s,t) = \frac{\partial^{\bm{k}}}{\partial \bm{\theta}^{\bm{k}}} \frac{\partial}{\partial s} \bm{F}(\bm{\theta}; s,t)\bigg|_{\bm{\theta} = \bm{0}}.
\end{align*} 
Differentiating both sides of \eqref{eq:mgf_pde}, we get by the generalized product rule    
\begin{align*}
&\frac{\partial^{\bm{k}}}{\partial \bm{\theta}^{\bm{k}}} \frac{\partial}{\partial s} \bm{F}(\bm{\theta}; s,t) \\
&\quad = r(s)\sum_{\ell=1}^n \sum_{\bm{y}\in \widetilde{S}(\bm{k})}\, \prod_{r = 1}^n {k_r \choose y_r}\!\left[\frac{\partial^{\bm{y}}}{\partial \bm{\theta}^{\bm{y}}}\theta_\ell\right] \frac{\partial^{\bm{k}-\bm{y}}}{\partial \bm{\theta}^{\bm{k}-\bm{y}}}\frac{\partial }{\partial \theta_\ell}\bm{F}(\bm{\theta}; s,t)  \\
 &\quad \!-\sum_{\bm{y}\in \widetilde{S}(\bm{k})} \, \prod_{r = 1}^n {k_r \choose y_r}  \frac{\partial^{\bm{y}}}{\partial \bm{\theta}^{\bm{y}}}\!\left[\bm{M}(s)\bullet \tuborg{e^{\ip{\bm{\theta}}{\bm{b}_{ij}(s)}}}_{i,j\in \J}+\sum_{\ell  = 1}^n \theta_\ell \bm{\Delta}(\bm{b}^\ell(s))\right]\! \frac{\partial^{\bm{k}-\bm{y}}}{\partial \bm{\theta}^{\bm{k}-\bm{y}}}  \bm{F}(\bm{\theta}; s,t).  \\ 
&\quad =r(s)\sum_{\ell=1}^n \left[\mathds{1}_{(\bm{y}=\bm{0})}\theta_\ell + k_\ell \mathds{1}_{(\bm{y}=\bm{e}_\ell)} \right] \frac{\partial^{\bm{k}-\bm{y}}}{\partial \bm{\theta}^{\bm{k}-\bm{y}}}\frac{\partial }{\partial \theta_\ell}\bm{F}(\bm{\theta}; s,t)  \\
&\quad -\!\sum_{\bm{y}\in \widetilde{S}({\bm{k}})} \, \prod_{r = 1}^n {k_r \choose y_r}\!\left[\bm{M}(s)\bullet \tuborg{\prod_{\ell = 1}^n b^\ell_{ij}(s)^{y_\ell} e^{\ip{\bm{\theta}}{\bm{b}_{ij}(s)}}}_{i,j\in \J}+\sum_{\ell  = 1}^n \mathds{1}_{(\bm{y}= \bm{e}_\ell)} \bm{\Delta}(\bm{b}^\ell(s))\right]\times \\ 
&\qquad\qquad\qquad\qquad\quad   \frac{\partial^{\bm{k}-\bm{y}}}{\partial \bm{\theta}^{\bm{k}-\bm{y}}}  \bm{F}(\bm{\theta}; s,t).
\end{align*}
In the last equality, we have used that $\prod_{r=1}^n \binom{k_r}{y_r} = k_\ell$ if $\bm{y} = \bm{e}_\ell$ along with $\frac{\partial^{\bm{y}}}{\partial \bm{\theta}^{\bm{y}}} \theta_\ell = \mathds{1}_{(\bm{y}=\bm{0})}\theta_\ell + \mathds{1}_{(\bm{y}=\bm{e}_\ell)}$ for $\ell=1,\ldots,n$.\ Now, evaluating in $\bm{\theta}=\bm{0}$, we get 
\begin{align*}
\frac{\partial}{\partial s} \bm{V}^{(\bm{k})}(s,t) &= r(s)\sum_{\ell=1}^n k_\ell \frac{\partial^{\bm{k}-\bm{e}_\ell}}{\partial \bm{\theta}^{\bm{k}-\bm{e}_\ell}}\frac{\partial }{\partial \theta_\ell}\bm{F}(\bm{\theta}; s,t)\bigg|_{\bm{\theta} = \bm{0}}\\
&\quad -\!\!\!\!\sum_{\bm{y}\in \widetilde{S}({\bm{k}})} \, \prod_{\ell = 1}^n {k_\ell \choose y_\ell}\!\left[\bm{M}(s)\bullet \tuborg{\prod_{\ell = 1}^n b^\ell_{ij}(s)^{y_\ell}}_{i,j\in \J} +\sum_{\ell  = 1}^n \mathds{1}_{(\bm{y}= \bm{e}_\ell)} \bm{\Delta}(\bm{b}^\ell(s))\right]\times \\ 
&\qquad\qquad\qquad\qquad\ \   \bm{V}^{(\bm{k}-\bm{y})}(s,t) \\
&=\bar{\bm{k}}\, r(s)\bm{V}^{(\bm{k})}(s,t) \\
&\quad -\!\!\!\!\sum_{\bm{y}\in \widetilde{S}({\bm{k}})} \prod_{\ell = 1}^n \!{k_\ell \choose y_\ell}\!\!\left[\bm{M}(s)\bullet B_1^{\bullet y_1}(s)\bullet\ldots\bullet B_n^{\bullet y_n}(s) +\sum_{\ell  = 1}^n \mathds{1}_{(\bm{y}= \bm{e}_\ell)} \bm{\Delta}(\bm{b}^\ell(s))\right]\!\times \\ 
&\qquad\qquad\qquad\qquad\ \   \bm{V}^{(\bm{k}-\bm{y})}(s,t).
\end{align*}
Splitting the last sum into terms corresponding to $\bm{y} = \bm{0}$ and $\bm{y}= \bm{e}_\ell$, for $\ell = 1,\ldots,n$, gives 
\begin{align*}
\frac{\partial}{\partial s} \bm{V}^{(\bm{k})}(s,t) &= \left(\bar{\bm{k}}\, r(s)\bm{I}-\bm{M}(s)\right)\!\bm{V}^{(\bm{k})}(s,t) - \sum_{\ell=1}^n k_\ell \bm{R}_\ell(s)\bm{V}^{(\bm{k}-\bm{e}_\ell)}(s,t) \\
&\quad - \sum_{\bm{y}\in S({\bm{k}})\atop \bm{y}\notin  E_n  }  \prod_{\ell = 1}^n {k_\ell \choose y_\ell}\bm{C}^{(\bm{y})}(s)\bm{V}^{(\bm{k}-\bm{y})}(s,t),
\end{align*}
the desired differential equation. The boundary condition follows from \eqref{eq:Vij} with the fact that $\bm{V}^{(\bm{0})}(t,t) = \bm{P}(t,t) = \bm{I}$. 
\end{proof}
\begin{remark}
The differential equation of Theorem  \ref{thm:diff_V} generalizes \cite[Theorem 6]{NorbergHigherOrder} to multivariate present values; for $n=1$ we have, supressing the superfluous dependency in $\ell$, that $S(k) = \tuborg{1,2,\ldots,k}$ and $E_1 = \tuborg{1}$, and so the differential equations reads
\begin{align*}
\frac{\partial}{\partial s} \bm{V}^{(k)}(s,t) &= \left(k\, r(s)\bm{I}-\bm{M}(s)\right)\!\bm{V}^{(k)}(s,t) - k \bm{R}(s)\bm{V}^{(k-1)}(s,t) \\
&\quad - \sum_{y=2}^k {k \choose y}\bm{C}^{(y)}(s)\bm{V}^{(k-y)}(s,t),
\end{align*}
with the terminal conditions $\bm{V}^{(k)}(t,t)=\mathds{1}_{(k=0)}\bm{I}$. \demoo 
\end{remark}
From the differential equation of the conditional partial moments, we immediately obtain differential equations  for the corresponding conditional moments $V_i^{(\bm{k})}(s,t)$ by multiplying the former with $\bm{1} = (1,1,\ldots,1)'$, as explained in Section \ref{sec:setup}.   
\begin{theorem}\label{thm:diff_Vi}
The conditional moments $V_i^{(\bm{k})}(\cdot,t)$ satisfies the backward differential equations given by, for $i\in \J$,
\begin{align} \label{eq:diff_Vi}
\begin{split} \frac{\partial }{\partial s} V_i^{(\bm{k})}(s,t) &= \left(\bm{\bar{k}}r(s)+\mu_{i\cdot}(s)\right)\! V_i^{(\bm{k})}(s,t) - \sum_{\ell=1}^n k_\ell \, b^\ell_i(s) V_i^{(\bm{k}-\bm{e}_\ell)}(s,t) \\ 
&\quad -\sum_{j:j\neq i}\mu_{ij}(s)\!\sum_{\bm{y}\in \widetilde{S}(\bm{k})}\,\prod_{\ell=1}^n \binom{k_\ell}{y_\ell}(b^\ell_{ij}(s))^{y_\ell}V_j^{(\bm{k}-\bm{y})}(s,t),
\end{split}  
\end{align}
with the terminal conditions $V_i^{(\bm{k})}(t,t) = \mathds{1}_{(\bm{k}=\bm{0})}$. 
\end{theorem}
\begin{proof}
Simply multiply the differential equation in Theorem \ref{thm:diff_V} with $\bm{1} = (1,1,\ldots,1)'$ on both sides and extract the $i$'th element of the vectors.
\end{proof}
\begin{remark}\label{remark:central_moment}
The differential equation of Theorem  \ref{thm:diff_Vi} generalizes \cite[(3.2)]{NorbergHigherOrder} to multivariate present values; for $n=1$ we have, supressing the superfluous dependency in $\ell$, that $S(k) = \tuborg{1,2,\ldots,k}$ and $E_1=\tuborg{1}$, and so the differential equations reads
\begin{align*}
\frac{\partial }{\partial s} V_i^{(k)}(s,t) &= \left(k r(s)+\mu_{i\cdot}(s)\right)\! V_i^{(k)}(s,t) - k  b_i(s) V_i^{(k-1)}(s,t) \\
&\quad -\sum_{j:j\neq i}\mu_{ij}(s)\sum_{y=0}^k \binom{k}{y}(b_{ij}(s))^{y}V_j^{(k-y)}(s,t),  
\end{align*}
with the terminal conditions $V_i^{(k)}(t,t) = \mathds{1}_{(k=0)}$. \demoo 
\end{remark}    
\begin{remark}
The conditional central moments
\begin{align*}
m_i^{(\bm{k})}(s,t) = \E\!\left[\left. \, \prod_{\ell = 1}^n \left(U_\ell(s,t) - V_i^{(\bm{e}_\ell)}(s,t)\right)^{k_\ell}\, \right| \, Z(s) = i\right]
\end{align*}
may be obtained from the conditional (non-central) moments $V_i^{(\bm{k})}(s,t)$ via an application of the multidimensional binomial formula, giving 
\begin{align}\label{eq:centralmoment}
m_i^{(\bm{k})}(s,t) = \sum_{\bm{y}\in \widetilde{S}(\bm{k})} \, \prod_{\ell = 1}^n (-1)^{k_\ell - y_\ell}\binom{k_\ell}{y_\ell}V_i^{(\bm{y})}(s,t)V_i^{(\bm{e}_\ell)}(s,t)^{k_\ell - y_\ell}.
\end{align} 
Thus, by first solving the differential equation of Theorem \ref{thm:diff_Vi} we are able to compute $m_i^{(\bm{k})}(s,t)$ via \eqref{eq:centralmoment}.\ Here, we may note that solving the differential equation of Theorem \ref{thm:diff_Vi} to compute the $\bm{k}$'th moment immediately gives $V_i^{(\bm{y})}(s,t)$ for all $\bm{y}\in \widetilde{S}(\bm{k})$; details regarding this is explained in Subsection \ref{subsec:prodint_moments}. 
\demoo 
\end{remark}
\begin{example}[Conditional covariance]\label{ex:prod_moment}
For $n=2$ and $\bm{k} = (1,1)$ we have that $\widetilde{S}(\bm{k}) = \tuborg{(0,0),(1,0), (0,1) , (1,1)}$, and so the differential equation for the conditional product moment $V_i^{(1,1)}(s,t)=\E\!\left[\left. U_1(s,t)U_2(s,t)\, \right| Z(s) = i\right]$ is given by
\begin{align*}
\frac{\partial}{\partial s} V_i^{(1,1)}(s,t) &= \left(2r(s)+\mu_{i\cdot}(s)\right)\!V_i^{(1,1)}(s,t) - b_i^1(s)V_i^{(0,1)}(s,t) - b_i^2(s)V_i^{(1,0 )}(s,t) \\
&\quad - \sum_{j:j\neq i}\mu_{ij}(s)\bigg(V_j^{(1,1)}(s,t)+b_{ij}^1(s)V_j^{( 0,1)}(s,t) + b_{ij}^2(s)V_j^{( 1,0)}(s,t) \\
&\qquad\qquad\qquad\qquad 
+b_{ij}^1(s) b_{ij}^2(s)
\bigg),  \\
V_i^{(1,1)}(t,t) &= 0,
\end{align*}
where $V_i^{(0,1)}$ and $V_i^{(1,0)}$ are the state-wise prospective reserves of the two payment processes, which can be calculated using Thiele's differential equation. From this, the conditional central moment 
$m_i^{(1,1)}(s,t)$ is given as
\begin{align}\nonumber
m_i^{(1,1)}(s,t) &= \E\!\left[\left. \left(U_1(s,t) - V_i^{(1,0)}(s,t)\right)\!\left(U_2(s,t) - V_i^{(0,1)}(s,t)\right) \, \right| \, Z(s) = i\right] \\ \label{eq:covariance}
&=V_i^{(1,1)}(s,t)-V_i^{(1,0)}(s,t)V_i^{(0,1)}(s,t),
\end{align}
which is the conditional covariance of $U_1(s,t)$ og $U_2(s,t)$ given $Z(s) = i$.\demo
\end{example}
Following up on the example, it turns out that it is possible to derive differential equations for the conditional covariance between two present values
in which the sum at risk used in the Thiele differential equations appears. The result is also presented without proof in \cite[Proposition 10.3]{asmussensteffensen}, and we are able to proof the result here through our results on the non-central moments.    

\begin{corollary}\label{cor:covariance}
The conditional covariance $m_i^{(\bm{e}_\ell+\bm{e}_m)}(s,t)$ between $U_\ell(s,t)$ and $U_m(s,t)$ given $Z(s) = i$ satisfies the system of backward differential equations given by, for $i\in \J$, 
\begin{align*}
\frac{\partial}{\partial s}m_i^{(\bm{e}_\ell+\bm{e}_m)}(s,t) &= 2r(s)m_i^{(\bm{e}_\ell+\bm{e}_m)}(s,t) \\
&\quad - \sum_{j:j\neq i} \mu_{ij}(s)\!\left(R_{ij}^\ell(s)R_{ij}^m(s)+m_j^{(\bm{e}_\ell+\bm{e}_m)}(s,t)-m_i^{(\bm{e}_\ell+\bm{e}_m)}(s,t)\right)\!,
\end{align*}
with terminal conditions $m_i^{(\bm{e}_\ell+\bm{e}_m)}(t,t) = 0$.\ Here, $R_{ij}^\ell$ denotes the sum at risk in the $\ell$'th payment process for the transition from state $i$ to $j$, $i\neq j$, given by
\begin{align*}
R_{ij}^\ell(s) = b_{ij}^\ell(s) + V_j^{(\bm{e}_\ell)}(s,t) - V_i^{(\bm{e}_\ell)}(s,t), 
\end{align*}
and analogously for $R_{ij}^m$. 
\end{corollary} 
\begin{proof}
Differentiating \eqref{eq:covariance} w.r.t.\ $s$ on both sides gives
\begin{align}\nonumber
\frac{\partial}{\partial s}m_i^{(\bm{e}_\ell+\bm{e}_m)}(s,t) &= \frac{\partial}{\partial s} V_i^{(\bm{e}_\ell + \bm{e}_m)}(s,t) -  V_i^{(\bm{e}_m)}(s,t)\frac{\partial}{\partial s} V_i^{(\bm{e}_\ell)}(s,t) \\  \label{eq:diff_covariance}
&\quad  -V_i^{(\bm{e}_\ell)}(s,t)\frac{\partial}{\partial s} V_i^{(\bm{e}_m)}(s,t)
\end{align}
Then insert the expression for $\frac{\partial}{\partial s} V_i^{(\bm{e}_\ell + \bm{e}_m)}(s,t)$ obtained in Example \ref{ex:prod_moment} as well as the Thiele differential equations for the state-wise prospective reserves, given by
\begin{align*}
\frac{\partial }{\partial s}V_i^{(\bm{e}_\ell)}(s,t) = r(s)V_i^{(\bm{e}_\ell)}(s,t) - b^\ell_i(s)-\sum_{j:j\neq i}\mu_{ij}(s)R_{ij}^\ell(s),
\end{align*}  
and analogously for $V_i^{(\bm{e}_m)}$.\ By gathering the relevant terms, we obtain the desired differential equation. The boundary condition follows directly from \eqref{eq:covariance}.  
\end{proof}
\begin{remark}
Note that the differential equation of Corollary \ref{cor:covariance} corresponds to a Thiele differential equation with interest rate $2r$, no sojourn payments and transition payments $R_{ij}^\ell R_{ij}^m$ from state $i$ to state $j$.\ Consequently, the conditional covariance $m_i^{(\bm{e}_\ell+\bm{e}_m)}$ has the representation
\begin{align*}
m_i^{(\bm{e}_\ell+\bm{e}_m)}(s,t) &= \E\!\left[\left. \int_s^t e^{-\int_s^u 2r(v)\dd v}\dd \widetilde{B}(u)\, \right| Z(s) = i  \right] \\
\dd \widetilde{B}(u)&= \sum_{i \in \J}\mathds{1}_{(Z(u-)=i)}\!\sum_{j:j\neq i} R_{ij}^\ell (u) R_{ij}^m (u)\dd N_{ij}(u). 
\end{align*}  
In the case of the conditional variance, i.e.\ $\ell = m$, these type of representations are known as Hattendorff's theorem, which has been formulated in a multi-state Markovian framework by \cite{Ramlau1988} using martingale techniques and further generalized in \cite{norberg1992}.\ The result here may thus reveal that the probabilistic structures leading to Hattendorff's theorem is not only limited to the variance, but can be carried over to the covariance. \demoo
\end{remark}
From the conditional covariances (and variances when $\ell=m$), we are then able to put up the conditional covariance matrix of $\bm{U}(s,t)$ given $Z(s) = i$, 
\begin{align}\label{eq:covariance_matrix}
\bm{\Sigma}_i(s,t) =  \left\{m_i^{(\bm{e}_\ell+\bm{e}_m)}\right\}_{\ell,m=1,\ldots,n}  
\end{align}   
as well as the corresponding correlation matrix
\begin{align}\label{eq:correlation_matrix}
\bm{\rho}_i(s,t) = \left\{\frac{m_i^{(\bm{e}_\ell+\bm{e}_m)}(s,t)}{\sqrt{m_i^{(\bm{e}_\ell+\bm{e}_\ell)}(s,t)m_i^{(\bm{e}_m+\bm{e}_m)}(s,t)}}\right\}_{\ell,m=1,\ldots,n}
\end{align}
which makes us able to analyse pair-wise dependence structures between the present values.  In particular, these may be used to approximate joint safety margins via multivariate CLT approximations as outlined in Section \ref{sec:motivation}.

\subsection{Product integral representation of conditional partial moments}\label{subsec:prodint_moments}
In this subsection, we derive a product integral representation for the conditional partial moments. This allows for the partial moments to be treated both theoretically and numerically as  objects of its own within the theory of product integrals, where all properties and results for these can be benefited from, see e.g.\  \cite{Johansen, GillJohansen} for theoretical properties and \cite{HeltonStucwisch} for numerical schemes. Specific to this problem, the product integral shall even demonstrate how one may compute all moments $\bm{V}^{(\bm{y})}$, $\bm{y}\in \widetilde{S}(\bm{k})$, at once.

From the differential equations for the conditional (partial) moments, we see a structure similar to those of a single payment stream, namely that one must use all lower order moments $\bm{y}^1,\ldots,\bm{y}^{|\bm{k}|-1}\in S(\bm{k})$ (each computed using Theorem \ref{thm:diff_V} or \ref{thm:diff_Vi}) to compute the $\bm{k}$'th moment. Since each of these need their corresponding lower order moments $S(\bm{y}^m)$, $m=1,\ldots,|\bm{k}|-1$, we are able to compute the $\bm{k}$'th moment starting from the $\bm{0}$'th moment, which is $\bm{V}^{(\bm{0})}(s,t) = \bm{P}(s,t)$ (or $V_i^{(\bm{0})}(s,t) = 1$), and then iteratively for each $m=1,\ldots,|\bm{k}|$ compute the $\bm{y}^m$'th moment using the lower order moments $\bm{y}^1,\ldots,\bm{y}^{m-1}$.\ In total, one must then solve $(|\bm{k}|+1)$ $J\times J$-dimensional systems of differential equations for $\bm{V}^{(\bm{k})}$, and $(|\bm{k}|+1)$ $J$-dimensional systems for $V_i^{(\bm{k})}$.

However, this approach relies on how we sort the vectors $\bm{y}^1,\ldots,\bm{y}^{|\bm{k}|}$, or, equivalently, the order at which we solve the differential equations, since we must ensure that the sets $S(\bm{y}^m)$ are increasing, i.e.\ that $S(\bm{y}^{m-1})\subseteq S(\bm{y}^{m})$ for all $m=1,\ldots,|\bm{k}|$, such that we actually are able to draw upon all lower order moments when calculating a given moment $\bm{y}^m$.\ Note that this is a trivial matter in the case of a single payment process, since for $n=1$ we have $S(k) = \tuborg{1,\ldots,k}$. 

It turns out that the key to obtain the result is to order the vectors $\bm{y}^1,\ldots,\bm{y}^{|\bm{k}|}$ in a way that fortunately is standard in most software packages.

\begin{assumption}
We assume that the lower order moments $S(\bm{k}) = \tuborg{\bm{y}^1,\ldots,\bm{y}^{|\bm{k}|}}$ are \underline{lexicographical ordered}, that is, they satisfy that for all $m' > m$, $m,m'\in \tuborg{1,\ldots,|\bm{k}|}$, there exist $u\in\tuborg{1,\ldots,n}$ (dependent on $m$ and $m'$), such that 
\begin{align*}
y_\ell ^{m} &= y_\ell ^{m'} \ \text{for all} \ \ell<u,  \ \text{and}  \\
y_u ^{m} &< y_u ^{m'}.
\end{align*}
\end{assumption}  
In other words, $\bm{y}^{m'}$ is strictly larger than $\bm{y}^{m}$ in the first entry where they differ (it may be smaller in the remaining entries). Details regarding this type of ordering and its properties are presented in Appendix \ref{sec:lexi}.\ 

Now, define the matrix $\bm{F}_{\bm{U}}^{(\bm{k})}(x)$, $x\geq 0$, as
\begin{align*}
  \begin{pmatrix}
  \bm{M}(x)-\bar{\bm{k}}r(x)\bm{I} & \bm{f}_1^{(\bm{k})}(x) & \bm{f}_2^{(\bm{k})}(x)  & \hdots &  \bm{f}_{|\bm{k}|}^{(\bm{k})}(x) \\[0.2 cm]
\bm{0} & \bm{M}(x)-\bar{\bm{y}}^{|\bm{k}|-1}r(x)\bm{I} & \bm{f}_1^{(\bm{y}^{|\bm{k}|-1})}(x) & \hdots & \bm{f}_{|\bm{k}|-1}^{(\bm{y}^{|\bm{k}|-1})}(x)  \\[0.2 cm]
\vdots & \vdots & \ddots & \vdots & \vdots   \\
\bm{0}   & \bm{0} & \bm{0} & \bm{M}(x)-\bar{\bm{y}}^{1}r(x)\bm{I} & \bm{f}_1^{(\bm{y}^1)}(x) 
\\[0.1 cm]
\bm{0} & \bm{0} &\bm{0} & \bm{0} & \bm{M}(x) 
\end{pmatrix}
\end{align*}
with $\bm{f}^{(\bm{y}^m)}_u(x)$, for $m=1,\ldots,|\bm{k}|$ and  $u=1,\ldots,m$, given by, setting $\bm{y}^{0} := \bm{0}$,
\begin{align}\label{eq:f}
\bm{f}_u^{(\bm{y}^m)}(x) = \begin{cases}
y^m_1 \bm{R}_1(x) \quad &\text{if} \ \bm{y}^{m} - \bm{y}^{m-u}=\bm{e}_1, \\
y^m_2 \bm{R}_2(x) \quad &\text{if} \ \bm{y}^{m} - \bm{y}^{m-u}=\bm{e}_2, \\
\vdots &\vdots \\
y^m_n \bm{R}_n(x) \quad &\text{if} \ \bm{y}^{m} - \bm{y}^{m-u}=\bm{e}_n, \\
\prod_{\ell = 1}^n \binom{y^m_\ell}{\xi_\ell} \bm{C}^{(\bm{\xi})}(x), &\text{if}  \ \bm{y}^{m} - \bm{y}^{m-u} = \bm{\xi} \in S({\bm{y}^m})\setminus{E_n}, \\
\bm{0} &\text{Otherwise}.
\end{cases} 
\end{align}
Here, one should think of the matrices $\tuborg{\bm{f}_u^{(\bm{y}^m)}(x)}_{u=1}^m$ as those needed to compute the $\bm{y}^m$'th moment, cf.\ Theorem \ref{thm:diff_V}, and so each block row of $\bm{F}_{\bm{U}}^{(\bm{k})}(x)$ corresponds to a calculation of a lower order moment. 
\begin{example}[Conditional covariance]\label{ex:prod_moment_prodint}
For the computation of conditional covariances, we have for $n=2$ and $\bm{k} = (1,1)$ that the matrix $\bm{F}_{\bm{U}}^{(1,1)}(x)$ reads as  
\begin{align*}
\bm{F}_{\bm{U}}^{(1,1)}(x) =  \begin{pmatrix}
\bm{M}(x)-2r(x)\bm{I} & \bm{R}_2(x) & \bm{R}_1(x) & \bm{C}^{(1,1)}(x) \\
\bm{0} & \bm{M}(x)-r(x)\bm{I} & \bm{0}  & \bm{R}_1(x) \\
\bm{0}  & \bm{0} & \bm{M}(x)-r(x)\bm{I}  & \bm{R}_2(x) \\
 \bm{0}  & \bm{0} & \bm{0} & \bm{M}(x)
\end{pmatrix},
\end{align*}
arising from the fact that the lexicographical ordering of $S((1,1))$ is given as $(0,1),(1,0)$, and $(1,1)$.  \demo
\end{example}
Now, let
\begin{equation}\label{eq:prodint_V}
 \bm{G}^{(\bm{k})}(s,t) = \prod_s^t \left(\bm{I}+\bm{F}_{\bm{U}}^{(\bm{k})}(x)\!\dd x\right)
 \end{equation} 
denote the product integral on the interval $(s,t]$ of the matrix function $\bm{F}^{(\bm{k})}_{\bm{U}}$. Since the matrix $\bm{F}_{\bm{U}}^{(\bm{k})}(x)$ is a $ \big(|\bm{k}|+1\big)\times \big(|\bm{k}|+1\big)$ block-partitioned matrix with blocks of sizes $J\times J$, we have that $\bm{G}^{(\bm{k})}(s,t)$ is of similar size, and we denote with $\bm{G}_{ij}^{(\bm{k})}(s,t)$ the $ij$'th block of  $\bm{G}^{(\bm{k})}(s,t)$ for $i,j\in\{1,\ldots |\bm{k}|+1\}$.\ The following result then demonstrates how all moments of up to order $\bm{k}$ are obtained through this single calculation of a product integral. 
\begin{theorem}\label{thm:matrix_V}
For each $i\in \tuborg{0,1,\ldots,|\bm{k}|}$ we have, for $m=0,1,\ldots,i$, 
\begin{equation}\label{eq:V_prodint1}
\bm{G}^{(\bm{k})}_{|\bm{k}|+1-i, |\bm{k}|+1-m}(s,t) = \mathds{1}_{\left(\bm{y}^i\geq \bm{y}^m\right) }\prod_{\ell = 1}^n \binom{y_\ell^i}{y_\ell^m}e^{-\bm{\bar{y}^m}\int_s^t r(v)\dd v}\bm{V}^{(\bm{y}^i-\bm{y}^m)}(s,t),
\end{equation}
which in particular gives the conditional $\bm{y}^i$'th moment for $m=0$:
\begin{equation}\label{eq:V_prodint2}
\bm{G}^{(\bm{k})}_{|\bm{k}|+1-i, |\bm{k}|+1}(s,t) = \bm{V}^{(\bm{y}^i)}(s,t).
\end{equation}
\end{theorem}
\begin{proof}
We shall mimic the proof of \cite[Theorem 3-4]{Bladt2020} and modify suitably to the present setup. By multiplying \eqref{eq:diff_Vij} with $\prod_t^s \!\left(\bm{I}+\!\left[\bm{M}(u)-\bar{\bm{k}}r(u)\bm{I}\right]\!\!\dd u\right)$ on both sides we get
\begin{align*}
&\frac{\partial}{\partial s}\!\left(\prod_t^s \!\left(\bm{I}+\!\left[\bm{M}(u)-\bar{\bm{k}}r(u)\bm{I}\right]\!\!\dd u\right)\! \bm{V}^{(\bm{k})}(s,t)\!\right) \\
&\quad = - \sum_{\ell=1}^n k_\ell \prod_t^s \!\left(\bm{I}+\!\left[\bm{M}(u)-\bar{\bm{k}}r(u)\bm{I}\right]\!\!\dd u\right)\!\bm{R}_\ell(s)\bm{V}^{(\bm{k}-\bm{e}_\ell)}(s,t) \\
&\quad\quad - \sum_{\bm{y}\in S({\bm{k}}) \atop \bm{y}\notin  E_n}   \prod_{\ell = 1}^n {k_\ell \choose y_\ell}\prod_t^s \!\left(\bm{I}+\!\left[\bm{M}(u)-\bar{\bm{k}}r(u)\bm{I}\right]\!\!\dd u\right)\!\bm{C}^{(\bm{y})}(s)\bm{V}^{(\bm{k}-\bm{y})}(s,t). 
\end{align*}
Integrating the equation yields 
\begin{align}\nonumber
\bm{V}^{(\bm{k})}(s,t) &= \sum_{\ell = 1}^n k_\ell \int_s^t \prod_s^x \!\left(\bm{I}+\!\left[\bm{M}(u)-\bar{\bm{k}}r(u)\bm{I}\right]\!\!\dd u\right)\! \bm{R}_\ell (x) \bm{V}^{(\bm{k}-\bm{e}_\ell)}(x,t)\dd x \\ 
\nonumber
&\quad + \sum_{\bm{y}\in S({\bm{k}}) \atop \bm{y}\notin  E_n}  \prod_{\ell = 1}^n {k_\ell \choose y_\ell}\int_s^t \prod_s^x \!\left(\bm{I}+\!\left[\bm{M}(u)-\bar{\bm{k}}r(u)\bm{I}\right]\!\!\dd u\right)\!\bm{C}^{(\bm{y})}(x)\bm{V}^{(\bm{k}-\bm{y})}(x,t) \dd x. 
\end{align}
We now provide an induction argument to verify the identity \eqref{eq:V_prodint1} claimed in the theorem using this integral equation. For $i=0$ and $m=0$ the identity is trivially true. So assume that the identity is true for some $i-1$,  $i\in \tuborg{1,\ldots,|\bm{k}|+1}$, and corresponding $m=0,1\ldots,i-1$. Then \cite[Lemma 1]{Bladt2020} gives us that
\begin{align*}
\bm{G}^{(\bm{k})}_{|\bm{k}|+1-i,|\bm{k}|+1-i}(s,t) = \prod_s^t \big(\bm{I}+\!\left[\bm{M}(x)-\bar{\bm{y}}^ir(x)\bm{I}\right]\!\!\dd x\big) = e^{-\bar{\bm{y}}^i\int_s^t r(v)\dd v}\bm{P}(s,t),
\end{align*} 
which is \eqref{eq:prodint_V} for $m=i$.\ From the Lemma, it also follows that for $m=0,1,\ldots,i-1$, 
\begin{align*}
\bm{G}^{(\bm{k})}_{|\bm{k}|+1-i,|\bm{k}|+1-m}(s,t) = \sum_{j=1}^{i-m} \!\int_s^t \! e^{-\bar{\bm{y}}^i\!\int_s^x r(v)\dd v}\bm{P}(s,x)\bm{f}_j^{(\bm{y}^i)}(x)\bm{G}^{(\bm{k})}_{|\bm{k}|+1-(i-j),|\bm{k}|+1-m}(x,t)\!\dd x
\end{align*}
From the induction hypothesis we have that
\begin{align*}
\bm{G}^{(\bm{k})}_{|\bm{k}|+1-(i-j),|\bm{k}|+1-m}(x,t) = \mathds{1}_{(\bm{y}^{i-j}\geq \bm{y}^m )}\prod_{\ell = 1}^n \binom{y_\ell^{i-j}}{y_\ell^m}e^{-\bm{\bar{y}^m}\int_x^t r(v)\dd v}\bm{V}^{(\bm{y}^{i-j}-\bm{y}^m)}(x,t),
\end{align*}
and furthermore, $\bm{f}_j^{(\bm{y}^i)}(x)$ may be written as
\begin{align*}
\bm{f}_j^{(\bm{y}^i)}(x) = \sum_{\ell = 1}^n \mathds{1}_{(\bm{y}^i - \bm{y}^{i-j} = \bm{e}_\ell)} y^i_\ell \bm{R}_\ell(x) + \sum_{\bm{\xi}\in S(\bm{y}^i) \atop \bm{\xi}\notin E_n}\mathds{1}_{(\bm{y}^i - \bm{y}^{i-j} = \bm{\xi})} \prod_{\ell = 1}^n \binom{y^i_\ell}{\xi_\ell} \bm{C}^{(\bm{\xi})}(x).  
\end{align*}
Note that when $\bm{y}^i - \bm{y}^{i-j} = \bm{\xi}$ for some $\bm{\xi}\in S(\bm{y}^i)$ (including the unit vectors) and some $j=1,\ldots,i-m$, it holds 
\begin{align*}
\binom{y_\ell^{i-j}}{y_\ell^m}\binom{y^i_\ell}{\xi_\ell} = \binom{y_\ell^i}{y_\ell^m}\binom{y_\ell^i - y_\ell^m}{\xi_\ell}
\end{align*}
for all $\ell=1,\ldots,n$.\ Thus, we now have 
\begin{align*}
&\bm{G}^{(\bm{k})}_{|\bm{k}|+1-i,|\bm{k}|+1-m}(s,t) \\
&\quad = \prod_{\ell = 1}^n \binom{y_\ell^i}{y_\ell^m}  \!\int_s^t \! e^{-\bar{\bm{y}}^i\!\int_s^x r(v)\dd v}\bm{P}(s,x)e^{-\bm{\bar{y}}^m\int_x^t r(v)\dd v}  \sum_{j=1}^{i-m}\mathds{1}_{(\bm{y}^{i-j}\geq \bm{y}^m )}\times \\
&\qquad\qquad\qquad\qquad\Bigg\{\sum_{\ell = 1}^n \mathds{1}_{(\bm{y}^i - \bm{y}^{i-j} = \bm{e}_\ell)} (y^i_\ell-y_\ell^m) \bm{R}_\ell(x)\bm{V}^{(\bm{y}^i-\bm{y}^m-\bm{e}_\ell)}(x,t) \\
&\qquad\qquad\qquad\qquad + \!\!\sum_{\bm{\xi}\in S(\bm{y}^i) \atop \bm{\xi}\notin E_n}\mathds{1}_{(\bm{y}^i - \bm{y}^{i-j} = \bm{\xi})} \prod_{\ell = 1}^n \binom{y^i_\ell-y_\ell^m}{\xi_\ell} \bm{C}^{(\bm{\xi})}(x)\bm{V}^{(\bm{y}^i-\bm{y}^m-\bm{\xi})}(x,t)\Bigg\}\!\dd x 
\end{align*}
Then observe that the discount factors can be factored as, for $x\in [s,t]$,  
\begin{align*}
e^{-\bar{\bm{y}}^i\!\int_s^x r(v)\dd v}e^{-\bm{\bar{y}}^m\int_x^t r(v)\dd v} = e^{-\bm{\bar{y}^m}\int_s^t r(v)\dd v}e^{-(\bm{\bar{y}}^i-\bm{\bar{y}^m})\int_s^x r(v)\dd v}.
\end{align*}
Furthermore, when $\bm{y}^{i-j}\geq \bm{y}^m$ and $\bm{y}^i-\bm{y}^{i-j}=\bm{\xi}$ for some $\bm{\xi}\in S(\bm{y}^i)$ and $j=1,\ldots,i-m$, we have that $\bm{\xi}\leq \bm{y}^i-\bm{y}^m$, and so the last sum can be carried out over $S(\bm{y}^i-\bm{y}^m)\setminus E_n$.\ Thus, we now have 
\begin{align*}
&\bm{G}^{(\bm{k})}_{|\bm{k}|+1-i,|\bm{k}|+1-m}(s,t) \\
&\quad = \mathds{1}_{(\bm{y}^i\geq \bm{y}^m)}\prod_{\ell = 1}^n \binom{y_\ell^i}{y_\ell^m}e^{-\bar{\bm{y}}^m\!\int_s^t r(v)\dd v}\times \\ 
&\qquad\qquad\qquad \!\int_s^t \! e^{-(\bar{\bm{y}}^i-\bar{\bm{y}}^m)\!\int_s^x r(v)\dd v}\bm{P}(s,x)\Bigg\{ \sum_{\ell = 1}^n  (y^i_\ell-y_\ell^m) \bm{R}_\ell(x)\bm{V}^{(\bm{y}^i-\bm{y}^m-\bm{e}_\ell)}(x,t) \\
&\qquad\qquad\qquad\qquad + \!\!\sum_{\bm{\xi}\in S(\bm{y}^i-\bm{y}^m) \atop \bm{\xi}\notin E_n}\prod_{\ell = 1}^n \binom{y^i_\ell-y_\ell^m}{\xi_\ell} \bm{C}^{(\bm{\xi})}(x)\bm{V}^{(\bm{y}^i-\bm{y}^m-\bm{\xi})}(x,t)\Bigg\}\!\dd x  \\
&=\mathds{1}_{\left(\bm{y}^i\geq \bm{y}^m\right) }\prod_{\ell = 1}^n \binom{y_\ell^i}{y_\ell^m}e^{-\bm{\bar{y}^m}\int_s^t r(v)\dd v}\bm{V}^{(\bm{y}^i-\bm{y}^m)}(s,t),
\end{align*}
as claimed. In the first equality, we have used that for all $\bm{\xi} \in S(\bm{y}^i - \bm{y}^m)$, \begin{align*}
\sum_{j=1}^{i-m}\mathds{1}_{(\bm{y}^i - \bm{y}^{i-j} = \bm{\xi})} = \mathds{1}_{(\bm{y}^i\geq \bm{y}^m)},
\end{align*}
which follows from Lemma \ref{thm:prop_lexi}. 
\end{proof}
The theorem gives us that the right-block column of $\bm{G}^{(\bm{k})}(s,t)$ contains all moments of order up to $\bm{k}$, including the $\bm{0}$'th moment, as follows:
\begin{align*}
\prod_s^t \left(\bm{I}+\bm{F}_{\bm{U}}^{(\bm{k})}(x)\dd x\right) = \begin{pmatrix}
\ast & \ast & \ast & \ast &  \hdots & \ast & \bm{V}^{(\bm{k})}(s,t) \\
\ast & \ast & \ast & \ast &  \hdots & \ast & \bm{V}^{(\bm{y}^{|\bm{k}|-1})}(s,t) \\
\ast & \ast & \ast & \ast &  \hdots & \ast & \bm{V}^{(\bm{y}^{|\bm{k}|-2})}(s,t) \\
\vdots & \vdots & \vdots & \vdots &  & \vdots   & \vdots \\
\ast & \ast & \ast & \ast &  \hdots & \ast & \bm{V}^{(\bm{y}^{1})}(s,t)  \\
\ast & \ast & \ast & \ast &  \hdots & \ast & \bm{P}(s,t)
\end{pmatrix},
\end{align*}
and so all lower order up to order $\bm{k}$ are obtained through this product integral calculation. From this, the conditional moments $V_i^{(\bm{k})}$ and corresponding central moments $m_i^{(\bm{k})}$ are obtained via \eqref{eq:Vi} and \eqref{eq:centralmoment}, respectively. 

\section{Numerical example}\label{sec:numerical}
In this section, we illustrate the methods presented in the previous sections in a numerical example of the motivating example considered in Section 2.

We consider a 40-year old male today at time 0 with retirement age 65, such that $T = 25$.\ We suppose that the valuation basis is taken to be the technical basis in the numerical example of \cite{buchardt2015}, which is given by the following:  
\begin{align*}
r(s) &= 0.01, \\
\mu_{01}(s) &= \left(0.0004+10^{4.54+0.06(s+40)-10} \right)\!\mathds{1}_{(s\leq 25)}, \\
\mu_{10}(s) &= \left(2.0058\cdot e^{-0.117(s+40)}\right)\!\mathds{1}_{(s\leq 25)}, \\
\mu_{02}(s) &= 0.0005+10^{5.88+0.038(s+40)-10}, \\
\mu_{12}(s) &=  \mu_{02}(s)\!\left(1+\mathds{1}_{(s\leq 25)}\right). 
\end{align*}
Some of the transition intensities are inspired by the Danish G82M technical basis. We examine the pair-wise covariance and correlation structures of the product combinations in the active state, i.e.\ the entries of the conditional covariance and correlation matrix given in \eqref{eq:covariance_matrix}--\eqref{eq:correlation_matrix} with $i = 0$;\ this is illustrated in Figure \ref{fig:variance_covariance}--\ref{fig:correlations}. The calculations are based on a numerical solution of the differential equation of Corollary \ref{cor:covariance} for the conditional covariance.  
\begin{figure}[H]
\centering
\includegraphics[scale=0.55]{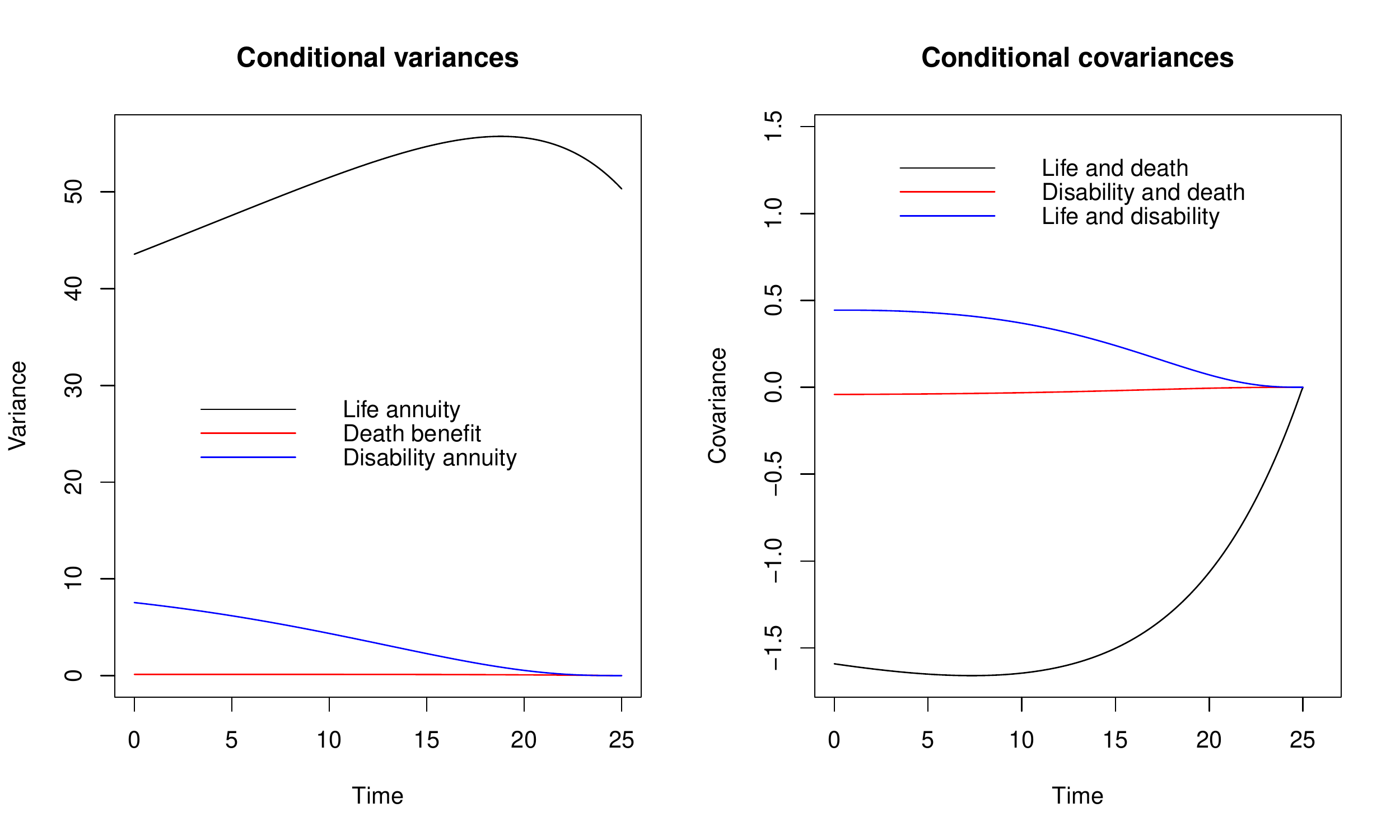}
\caption{Pair-wise conditional covariances $[0,25]\in s\mapsto \Sigma_0(s,70)_{\ell m}$, $\ell,m\in \tuborg{1,2,3}$, $\ell \leq m$, between the three products until retirement conditional on the insured being active.}
\label{fig:variance_covariance}
\end{figure}
\begin{figure}[H]
\centering
\includegraphics[scale=0.55]{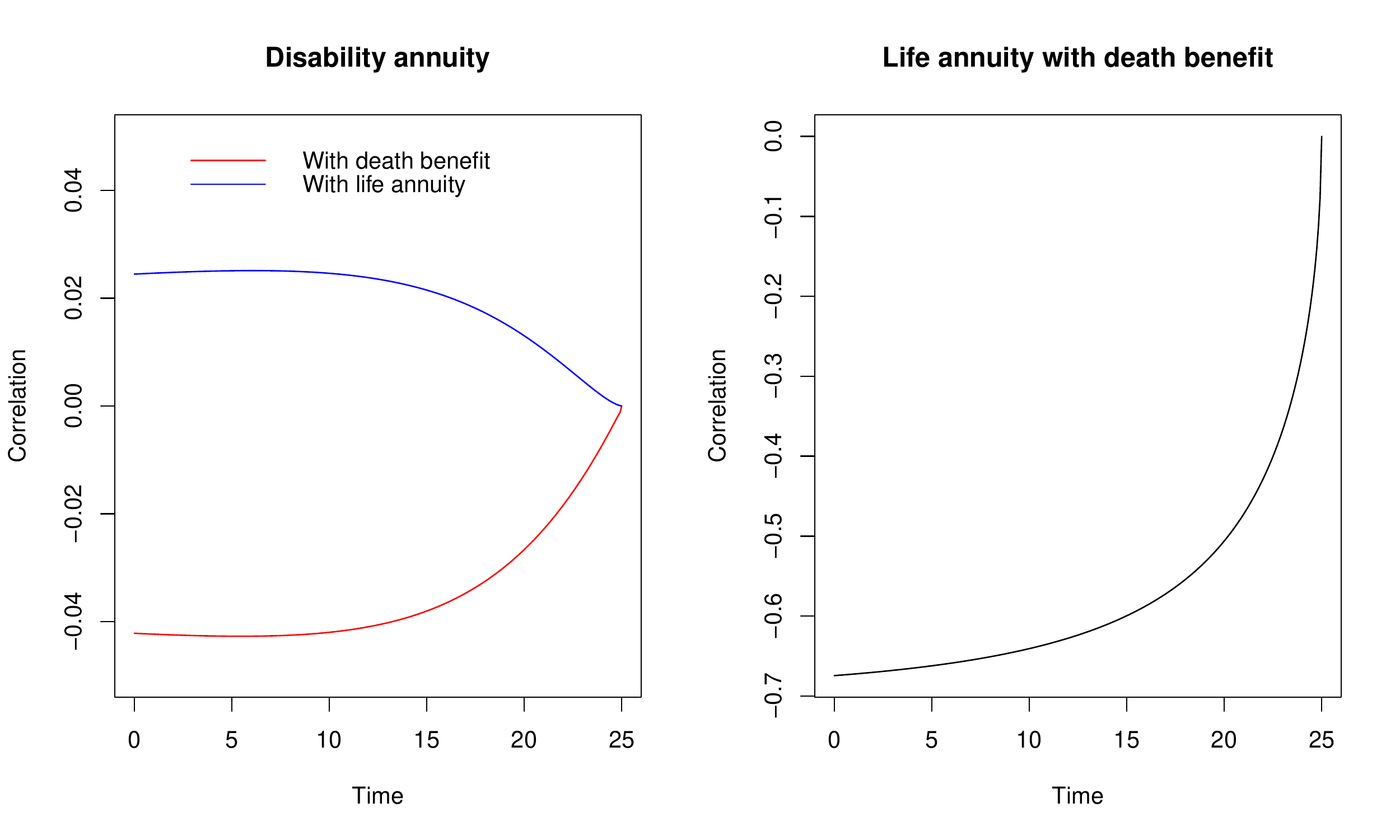}
\caption{Pair-wise conditional correlations $[0,25]\in s\mapsto \rho_0(s,70)_{\ell m}$, $\ell,m\in \tuborg{1,2,3}$, $\ell< m$, between the three products until retirement conditional on the insured being active.}
\label{fig:correlations}
\end{figure}

We see that the biggest dependence is between the life annuity and death benefit, which is was expected, but we also see a similar (but opposite) dependence structure between the disability annuity and death benefit respectively life annuity. The death benefit seem to be slightly more correlated with the disability annuity than the life annuity is, which is the most non-trivial observation encountered here. The calculated covariances can further be used to carry out the CLT approximation of the joint distribution outlined in Section \ref{sec:motivation}, which we refrain from doing here.

\section*{Acknowledgments}
I should like to thank my PhD supervisors Mogens Bladt and Mogens Steffensen for valuable suggestions and fruitful discussions during this work.

\bibliography{refsSmall}{}

\begin{thebibliography}{10}

\bibitem{AdekambiChristiansen2017}
F.~Adékambi and M.C. Christiansen.
\newblock Integral and differential equations for the moments of multistate
  models in health insurance.
\newblock {\em Scandinavian Actuarial Journal}, 2017(1):29--50, 2017.

\bibitem{AdekambiChristiansen2020}
F.~Adékambi and M.C. Christiansen.
\newblock Probability distributions of multi-states insurance models under
  semi-markov assumptions.
\newblock {\em Markov Processes and Related Fields}, 26(3):517--534, 2020.

\bibitem{asmussensteffensen}
S.~Asmussen and M.~Steffensen.
\newblock {\em Risk and Insurance}.
\newblock Series on Probability Theory and Stochastic Modelling. Springer,
  2020.

\bibitem{Bladt2020}
M.~Bladt, S.~Asmussen, and M.~Steffensen.
\newblock Matrix representations of life insurance payments.
\newblock {\em European Actuarial Journal}, pages 1--39, 01 2020.

\bibitem{buchardt2015}
K.~Buchardt and T.~Møller.
\newblock Life {I}nsurance {C}ash {F}lows with {P}olicyholder {B}ehavior.
\newblock {\em Risks}, 3(3):290–317, 2015.

\bibitem{christiansen2013}
M.C. Christiansen.
\newblock Safety margins for unsystematic biometric risk in life and health
  insurance.
\newblock {\em Scandinavian Actuarial Journal}, 2013(4):286--323, 2013.

\bibitem{coxlin}
S.H. Cox and Y.~Lin.
\newblock Natural hedging of life and annuity mortality risks.
\newblock {\em North American actuarial journal}, 11(3):1--15, 2007.

\bibitem{GillJohansen}
R.D. Gill and S.~Johansen.
\newblock A survey of product-integration with a view toward application in
  survival analysis.
\newblock {\em The Annals of statistics}, 18(4):1501--1555, 1990.

\bibitem{HeltonStucwisch}
J.~Helton and S.~Stuckwisch.
\newblock Numerical approximation of product integrals.
\newblock {\em Journal of mathematical analysis and applications},
  56(2):410--437, 1976.

\bibitem{HesselagerNorberg}
O.~Hesselager and R.~Norberg.
\newblock On probability distributions of present values in life insurance.
\newblock {\em Insurance, mathematics and economics}, 18(1):35--42, 1996.

\bibitem{hoem69}
J.M. Hoem.
\newblock Markov chain models in life insurance.
\newblock {\em Bl{\"a}tter der DGVFM}, 9:91--107, 1969.

\bibitem{hoemaalen}
J.M. Hoem and O.O. Aalen.
\newblock Actuarial values of payment streams.
\newblock {\em Scandinavian Actuarial Journal}, 1978(1):38--47, 1978.

\bibitem{Johansen}
S.~Johansen.
\newblock Product integrals and markov processes.
\newblock {\em CWI Newsletter}, (12):3--13, 1986.

\bibitem{italian}
S.~Levantesi and M.~Menzietti.
\newblock Natural hedging in long-term care insurance.
\newblock {\em ASTIN Bulletin}, 48(1):233–274, 2018.

\bibitem{milbrodtstracke1997}
H.~Milbrodt and A.~Stracke.
\newblock Markov models and thiele's integral equations for the prospective
  reserve.
\newblock {\em Insurance: Mathematics and Economics}, 19(3):187--235, 1997.

\bibitem{norberg1991}
R.~Norberg.
\newblock Reserves in {L}ife and {P}ension {I}nsurance.
\newblock {\em Scandinavian Actuarial Journal}, 1991:3--24, 1991.

\bibitem{norberg1992}
R.~Norberg.
\newblock Hattendorff's theorem and {T}hiele's differential equation
  generalized.
\newblock {\em Scandinavian Actuarial Journal}, 1992(1):2--14, 1992.

\bibitem{NorbergHigherOrder}
R.~Norberg.
\newblock Differential equations for moments of present values in life
  insurance.
\newblock {\em Insurance: Mathematics and Economics}, 17(2):171--180, 1995.

\bibitem{Norberg1995}
R.~Norberg.
\newblock A time-continuous {M}arkov chain interest model with applications to
  insurance.
\newblock {\em Applied Stochastic Models and Data Analysis}, 11(3):245--256,
  1995.

\bibitem{NorbergBonus2}
R.~Norberg.
\newblock A theory of bonus in life insurance.
\newblock {\em Finance and Stochastics}, 3(4):373--390, 1999.

\bibitem{Ramlau1988}
H.~Ramlau-Hansen.
\newblock Hattendorff's theorem: A markov chain and counting process approach.
\newblock {\em Scandinavian Actuarial Journal}, 1988(1-3):143--156, 1988.

\bibitem{Wolthuis1992}
H.~Wolthuis.
\newblock Prospective and retrospective premium reserves.
\newblock {\em Blätter der DGVFM}, 20(3):317--327, 1992.

\end{thebibliography}
\bibliographystyle{plain}

 \appendix

\section{Lexicographical ordering}\label{sec:lexi}
Let $\bm{k} = (k_1,\ldots,k_n)\in \N_0^n$ be a $n$-dimensional vector of natural numbers including zero. We think of this vector as representing some multivariate higher order moment we wish to calculate for the multivariate present value $\bm{U}$.\ Specific to this, we need to consider all combinations of lower order moments represented by the elements of $S(\bm{k})$; recall \eqref{eq:Sk} for its definition. Write these elements as
\begin{equation*}
S({\bm{k}}) =  \left\{\bm{y}^{1},\ldots,\bm{y}^{|\bm{k}|}
\right\}. 
\end{equation*}
The vectors $\bm{y}^{1},\ldots,\bm{y}^{|\bm{k}|}$ are then said to be \textit{lexicographical ordered} if for all $m' > m$, there exist $u$ (dependent on $m$ and $m'$) such that 
\begin{align*}
y_\ell ^{m} &= y_\ell ^{m'} \ \text{for all} \ \ell<u,  \ \text{and}  \\
y_u ^{m} &< y_u ^{m'}.
\end{align*}
In other words $y^{m'}$ must be strictly larger than $y^{m}$ in the first entry where they differ (it may be smaller in the remaining entries), and one commonly writes $\bm{y}^{m}<_{\text{lex}} \bm{y}^{m'}$.\ An illustration of the ordering is presented in Table \ref{tab:lexi_k}.  
\begin{table}[H]
\centering
\begin{tabular}{cccc}
$\bm{y}^1 = \left(0,\ldots,0,1\right)'$ & $\bm{y}^{k_n+1} = \left(0,\ldots,1,0\right)$ & & \!\!$\bm{y}^{|\bm{k}|-k_n} = \left(k_1,\ldots,k_{n-1},0\right)'$  \\
$\bm{y}^2 = \left(0,\ldots,0,2\right)'$ & $\bm{y}^{k_n+2} = \left(0,\ldots,1,1\right)$ &  & \!\!$\bm{y}^{|\bm{k}|-k_n+1} = \left(k_1,\ldots,k_{n-1},1\right)'$\\
\vdots & \vdots & $\cdots$  & \vdots \\
$\bm{y}^{k_n} = \left(0,\ldots,0,k_n\right)'$ & $\bm{y}^{2(k_n+1)-1} = \left(0,\ldots,1,k_n\right)'$ &  &  \!\!$\bm{y}^{|\bm{k}|} = \left(k_1,\ldots,k_n\right)'$ 
\end{tabular}
\caption{Illustration of lexicographical ordering of vectors in $S(\bm{k})$.} 
\label{tab:lexi_k}
\end{table}
The following result on this type of ordering provides the foundation to successively compute multivariate higher order moments based on already computed lower order moments.  
\begin{theorem}\label{thm:prop_lexi}
Assume that $\bm{y}^{1},\ldots,\bm{y}^{|\bm{k}|}$ are lexicographical ordered.\ Then for all  $i=1,\ldots,|\bm{k}|$ and $m=0,1,\ldots,i-1$ such that $\bm{y}^i\geq \bm{y}^m$, we have 
\begin{align*}
S(\bm{y}^i-\bm{y}^m)\subseteq \bigcup_{j=m}^{i-1} \left\{\bm{y}^i-\bm{y}^{j}\right\}. 
\end{align*} 
\end{theorem}
\begin{sproof}
We show the result for the two-dimensional case $n=2$, and the generalization to higher dimension then follows using same (but notationally cumbersome) techniques. Let $i=1,\ldots,|\bm{k}|$ be given. From the illustration of the lexicographic ordering of vectors in $S(\bm{k})$ shown in Table \ref{tab:lexi_k}, we may realize that the vectors take the form 
\begin{align*}
\bm{y}^i = \left(a, i-a(k_2+1)\right)', \quad \text{if} \ i\in [a(k_2+1), (a+1)(k_2+1)-1] 
\end{align*}
where $a=1,\ldots,k_1$.\ This can be verified by exploiting the structure of which the ordering is carried out. Now let $m=0,1\ldots,i-1$ be given such that $\bm{y}^i\geq \bm{y}^m$.\ Then similarly there exist $b=0,1,\ldots,k_1$ such that $m\in [b(k_2+1), (b+1)(k_2+1)-1]$, which then gives 
\begin{align}\label{eq:y_i-y_m}
\bm{y}^i - \bm{y}^m = \left(a-b, i-m-(a-b)(k_2+1)\right)',
\end{align}
By similar reasonings, we have for $j=m,\ldots,i-1$ there exist $c=1,\ldots,k_1$ such that $j\in [c(k_2+1), (c+1)(k_2+1)-1]$, giving
\begin{align*}
\bm{y}^i - \bm{y}^j = \left(a-c, i-j-(a-c)(k_2+1)\right)'.
\end{align*} 
Due to the indices having the order $m\leq j\leq i$, we have that $b\leq c \leq a$. Consequently, the set $\bigcup_{j=m}^{i-1} \left\{\bm{y}^i-\bm{y}^{j}\right\}$ consists of vectors on the form $\left(a-c, i-j-(a-c)(k_2+1)\right)'$ where $j$ varies on $m,\ldots,i-1$ and $c$ varies on $b,b+1,\ldots,a$. 

Now take $\bm{\xi} \in S(\bm{y}^i - \bm{y}^m)$.\ Since $\bm{\xi}\leq \bm{y}^i - \bm{y}^m$, it follows from \eqref{eq:y_i-y_m} that there exist $d\leq a-b$ and $r\leq i-m-(a-b)(k_2+1)$, $d,r\in \N_0$, such that 
\begin{align*}
\bm{\xi} &= \left(a-(b+d), i-(m+r)-(a-b)(k_2+1)\right)' \\
         &= \left(a-(b+d), i-(m+r+d(k_2+1))-(a-(b+d))(k_2+1)\right)',
\end{align*}  
where we have added and subtracted the term $d(k_2+1)$ in the second coordinate to obtain the last equality. We now see that $\bm{\xi}$ is on the form $\left(a-c, i-j-(a-c)(k_2+1)\right)'$ with $c = b+d$ and $j = m+r+d(k_2+1)$, and so we have $\bm{\xi}\in \bigcup_{j=m}^{i-1} \left\{\bm{y}^i-\bm{y}^{j}\right\}$, as claimed. This concludes the sketch of the proof.     
\end{sproof}

\end{document}